\documentclass[11pt]{article}
\usepackage[english]{babel}
\usepackage{graphicx}
\usepackage[sort,numbers]{natbib}
\usepackage{authblk}
\usepackage{framed}
\usepackage[normalem]{ulem}
\usepackage{amsmath,multirow}
\usepackage[ruled,vlined,linesnumbered]{algorithm2e}
\usepackage{amsthm,thm-restate,thmtools}
\usepackage{amssymb}
\usepackage{amsfonts}
\usepackage{enumerate}
\usepackage{etoolbox}
\usepackage{subcaption}
\usepackage{soul}
\usepackage{booktabs}
\usepackage{pifont}
\usepackage{palatino,mathpazo,microtype}
\usepackage[margin=15pt,font=small,labelfont={bf}]{caption} 

\usepackage{hyperref}
\hypersetup{
     colorlinks   = true,
     linkcolor    = red, 
     urlcolor     = teal, 
	 citecolor    = blue 
}

\usepackage[capitalise,noabbrev]{cleveref}
\usepackage{tabularx}

\usepackage{tikz}
\usetikzlibrary{decorations,arrows,petri,topaths,backgrounds,shapes,positioning,fit,calc,decorations.pathreplacing,patterns,intersections,decorations.pathmorphing,matrix,angles,quotes}

\tikzstyle{thickline} = [line width=1.8pt]
\tikzstyle{gl} = [draw, gray]
\tikzstyle{pl} = [draw, orange]
\tikzstyle{ml} = [draw, blue]
\tikzstyle{sett} = [draw, thick, purple]
\tikzstyle{ele} = [draw, thick, green!70!black]
\tikzset{room/.style={inner sep=2pt,minimum size=3pt, draw=black,fill=white,rectangle}}
\tikzset{agent/.style={inner sep=1pt,minimum size=2pt, draw=black,fill=white,circle}}

\usepackage{mathtools}
\usepackage[top=1 in,bottom=1in, left=1 in, right=1 in]{geometry}
\usepackage{xcolor,xfrac}
\usepackage{url}
\usepackage{enumitem}
\usepackage{tabularx}
\usepackage{todonotes}
\usepackage{algpseudocode}
\algdef{SE}[PROCEDURE]{Procedure}{EndProcedure}%
   [2]{\algorithmicprocedure\ \textproc{#1}\ifthenelse{\equal{#2}{}}{}{(#2)}}%
   {\algorithmicend\ \algorithmicprocedure}%

\newtheorem{theorem}{Theorem}[section]
\newtheorem{proposition}[theorem]{Proposition}
\newtheorem{lemma}[theorem]{Lemma}
\newtheorem{corollary}[theorem]{Corollary}
\newtheorem{definition}{Definition}
\newtheorem{observation}{Observation}[section]
\newtheorem{remark}[observation]{Remark}
\newtheorem*{example}{Example}

\DeclareMathOperator*{\argmax}{arg\,max}

\newcommand{\R}{\mathcal{R}}
\newcommand{\N}{\mathcal{N}}

\newcommand{\enn}{\hat{n}}
\newcommand{\emm}{\hat{m}}

\renewcommand{\emptyset}{\varnothing}

\newcommand{\opt}{\ensuremath{\mu^*}\xspace}
\newcommand{\xmark}{\ding{55}}

\newcommand{\USW}[1]{\ifstrempty{#1}{\textrm{\textup{SW}}}{#1\textrm{\textup{-SW{}}}}}

\newcommand{\HH}[1]{{\color{magenta}{HH: }{#1} \color{red}}}
\newcommand{\SR}[1]{{\color{blue!70!black}{SR:#1}}}
\newcommand{\SN}[1]{{\color{pink!50!red}{SN: }{#1} }}
\title{Strategyproof Matching of Roommates and Rooms}

\author[1]{Hadi Hosseini}
\author[2]{Shivika Narang}
\author[3]{Sanjukta Roy
}
\affil[1]{Pennsylvania State University}
\affil[2]{University of New South Wales}
\affil[3]{ISI Kolkata}
\affil[ ]{\texttt {hadi@psu.edu, s.narang@unsw.edu.au, sanjukta@isical.ac.in}}
\date{}

\begin{document}
\maketitle

\begin{abstract}
\noindent We initiate the study of matching roommates and rooms wherein the preferences of agents over other agents and rooms are complementary and represented by Leontief utilities. In this setting, $2n$ agents must be paired up and assigned to $n$ rooms. Each agent has cardinal valuations over the rooms as well as compatibility values over all other agents. Under Leontief preferences, an agent’s utility for a matching is the minimum of the two values. We focus on the tradeoff between maximizing utilitarian social welfare and strategyproofness.
Our main result shows that---in a stark contrast to the additive case---under binary Leontief utilities, there exist strategyproof mechanisms that maximize the social welfare. 
We further devise a strategyproof mechanism that implements such a welfare maximizing algorithm and is parameterized by the number of agents.
Along the way, we highlight several possibility and impossibility results, and give upper bounds and lower bounds for welfare with or without strategyproofness.
\end{abstract}

\section{Introduction}
The roommates matching problem is concerned with pairing up $2n$ agents (roommates) and assigning them to $n$ rooms.
This problem originates from the seminal paper by Gale and Shapley---and discussed by Knuth---as a generalization of the Stable Marriage problem \citep{gale1962college}. 
The goal is finding a stable disjoint pairing of agents such that no pair of agents would prefer each other to their current match.
It has inspired a large body of work on roommate matching \citep{IRVING1985577,knuth1976mariages, teo1998geometry, aziz2013stable},
three dimensional matching \citep{mckay2021three,BreCheFinNie2020-spscSM-jaamas,lam2019existence}, and coalition formation \citep{knuth1976mariages,DG1980Hedonic}.

Even though the problem was mainly inspired by the college dormitory assignment, the majority of the works on roommate matchings study stability where agents have preferences over other agents while being agnostic to their rooms. 
The 3-Dimensional extensions of this problem \citep{mckay2021three} primarily focus on finding stable coalitions of size three and do not consider preferences over different types of (often) complementary entities (e.g. roommates and rooms).
This is particularly an important issue in common online marketplaces for short and long-term stays such as Airbnb, Roomi\footnote{Roomi: https://roomiapp.com/}, and Roomsurf\footnote{Roomsurf: https://www.roomsurf.com/}, where agents can specify their preferences over rooms as well as roommates.

\paragraph{Preference Models.} 
Within three-dimensional matchings (as well as hedonic games), additivity of preferences is a common assumption. Additive separable preferences are derived from the summation of cardinal values that each agent assigns to every other agent \citep{mckay2021three,aziz2011optimal}.
Similarly, some recent work on matching roommates and rooms considers additive utilities over rooms and roommates \citep{chan2016assignment, huzhang2017online, li2019room, gan2020fair}.
Additive preferences model choice substitutes, i.e. receiving a less desirable item (e.g. a room) can be compensated by receiving a more desirable option (e.g. a roommate). 

While this is plausible when partitioning agents into coalitions, it does not always provide a sufficiently meaningful way to compare choices, particularly when outcomes involve \textit{complementary} alternatives. 
When matching roommates and rooms, for a person who values proximity to their classes, the utility derived from living with a close friend cannot necessarily \textit{substitute} for being assigned a room in a remote corner of the campus. 

\paragraph{Complementary Preferences via Leontief.} A compelling alternative model for capturing complementarity is the Leontief utility model, which is common in production economics \citep{leontief1965structure} as well as resource allocation and scheduling \citep{ghodsi2011dominant,parkes2015beyond,branzei2015characterization,dolev2012no}. Under Leontief utilities, different types of items or alternatives may be complements such as burgers and buns, pen and paper, someone's room and their roommate, etc. More broadly, in group projects, students are only able to produce good quality work if they get along with their partners and have some interest in their topic.  
In cloud computing, the allocation of computational resources (e.g. CPU and RAM), only one unit of a task requiring 2 CPUs and 3 GB of RAM can be completed with access to 4 CPUs and 4 GB of RAM (assuming indivisibility). 

In this paper, we consider the problem of matching roommates and rooms where the preferences of agents over other agents and rooms are given by cardinal values. 
The two fundamental notions in game theory and social choice are economic efficiency---measured by social welfare---and strategyproofness---no agent can obtain a more preferred outcome by manipulating its preferences.
%
%
Prior works in this space largely explored the objective of maximizing social welfare---as a desirable notion for economic efficiency---under additive utilities without considering strategyproofness.

We depart from prior works by (i) assuming that preferences are represented by Leontief utilities, and (ii) primarily focusing on designing strategyproof mechanisms that maximize the social welfare. 
\citet{gibbard1973manipulation} and 
\citet{satterthwaite1975strategy}'s seminal works showed that every strategyproof social choice function is either dictatorial or imposing (with more than two alternatives). 
Similarly, the incompatibility between strategyproofness and economic efficiency in several object/resource allocation settings \cite{zhou1990conjecture,hylland1979efficient} persisted, which motivated a variety of approaches to bypass the negative results by, for example, restricting the domain (e.g. preferences or outcomes) or sacrificing/weakening the efficiency requirements.
We study the interplay between strategyproofness and social welfare in the roommate-room matching setting.





\subsection{Our Contribution}
We focus on the problem of matching roommates and rooms when agents have arbitrary general valuations over each other and rooms, as well as when those valuations are restricted to binary values ($0/1$).
Our main results involve binary valuations, which give rise to complex preferences in this domains with multi-dimensional choices.
From the conceptual perspective, binary valuations represent practical applications in object or resource allocation (e.g. dichotomous preferences \cite{bogomolnaia2004random,bogomolnaia2005collective}) or voting theory (e.g. approval preferences \cite{elkind2023justifying,lu2024approval}). 
For instance, a student may approve/disapprove a potential roommate based on factors such as smoking habits; or she may have a value of zero for room options without dedicated bathrooms. 

From the theoretical perspective, binary valuations pose many technical challenges in achieving axiomatic results when dealing with strategyproofness  (e.g. stable matchings \cite{liu2023strategyproof}) and in developing computationally tractable algorithms (e.g. fair division \cite{halpern2020fair}).
%
Table~\ref{tab:results} summarizes our main results under Leontief utilities with general or binary valuations.

\begin{table}[t]
\crefname{lemma}{Lem.}{Lem.}
\crefname{theorem}{Thm.}{Thm.}
\crefname{proposition}{Prop.}{Prop.}
\crefname{corollary}{Cor.}{Cor.}
\centering
{
\begin{tabular}{@{}llllll@{}}
\toprule  
                & \multicolumn{2}{c}{\textbf{Leontief}}& \multicolumn{2}{c}{\textbf{Additive}}\\ 
                &  \textbf{General}            & \textbf{Binary}        & \textbf{General}              & \textbf{Binary}                        \\ \midrule
 \textbf{\USW{}}  &  APX-h$^\dagger$ (\cref{thm:SW-hardL})    & APX-h$^\dagger$ (\cref{thm:SW-hardL}) & NP-h$^{*}$ & NP-h$^{*}$  \\  
\textbf{SP Upper Bound} & \xmark{} (\cref{thm:no-alpha-approx}) & \textbf{1  (\cref{thm:spmech2})} &  \xmark{} (\cref{thm:no-alpha-approx}) & $\sfrac{2}{3}$ (\cref{thm:addsp23}) \\
\textbf{SP (Polytime)}  & \xmark{} (\cref{thm:no-alpha-approx}) & $\sfrac{1}{3}$ (\cref{thm:greedy}) & \xmark{} (\cref{thm:no-alpha-approx}) & $\sfrac{1}{7}$  (\cref{thm:sdadd}) \\
\bottomrule
\end{tabular}%
}

\caption{Summary of the tradeoff between social welfare (\USW{}) and strategyproofness (SP). Here \xmark{} denotes that no strategyproof  mechanism can give any non-zero approximation of \USW{}. APX-h  and NP-h denote APX and NP-hardness, resp. $^*$ denotes a result by \cite{chan2016assignment}  and $^{\dagger}$ indicates bounds that hold even for binary and symmetric valuations. 
}
\label{tab:results}
\end{table}

\paragraph{Social Welfare.}
We investigate various approaches for finding maximal matchings in the roommate matching problems, and analyze their welfare bounds. We develop a greedy algorithm for finding a maximal matching and show that its satisfies \USW{\sfrac{1}{3}}, resembling the same bound for maximal 3-dimensional matching.
We then show that under Leontief utilities, finding a social welfare (\USW{}) maximizing matching is APX-hard even when the valuations are binary and symmetric (\cref{thm:SW-hardL}).\footnote{Preferences are symmetric if Alice likes Bob if and only if that Bob likes Alice.}  
We then show that there exists an $0.559$-approximation algorithm for any type of utility function (\cref{prop:setpackbound}).

\paragraph{Strategyproof Mechanisms.}
When requiring strategyproofness, we show that no strategyproof mechanism can satisfy any approximation of social welfare, with arbitrary valuations over rooms or roommates. This negative result holds for both Leontief and additive utilities (\cref{thm:no-alpha-approx}), and raises the question of whether it is possible to escape the impossibility by restricting the valuation domain. 
For additive utilities with binary valuations, we prove an upper bound on the social welfare of any strategyproof mechanism. In particular, in \cref{thm:addsp23} and \cref{cor:addspsymm}, we show that no strategyproof mechanism can satisfy \USW{$\alpha$} for any $\alpha > 2/3$ for binary but not symmetric preferences  or $\alpha > 3/4$ for binary and symmetric preferences (\cref{thm:addsp23}).

Our main result shows that---in stark contrast with its additive counterpart---under Leontief utilities with binary valuations, there exists a strategyproof mechanism that maximizes social welfare (\cref{prop:spmech1}). 
We develop a novel mechanism that repeatedly reduces the welfare set by `shrinking' the potential strategy space. We show that this mechanism runs in time $O^*(n^{2n})$.\footnote{The $O^*$ notation hides factors polynomial in input size.}

\paragraph{Computation and Algorithms.}
Our welfare maximizing strategyproof mechanism runs in exponential time; given our intractability result in \cref{thm:SW-hardL}, no polynomial-time mechanism can be expected, unless P=NP.
Nonetheless, we develop a Fixed Parameter Tractable (FPT) algorithm that implements our strategyproof mechanism while maximizing the social welfare. This algorithm is parameterized by the number of agents which runs in time $O^*(c^n)$ where $c>0$ is a constant.\footnote{FPT parameterized by $k$ implies it runs in time $f(k)n^{O(1)}$ for some computable function $f$.}

Finally, we explore the welfare bound of various strategyproof mechanisms that run in polynomial time.
In particular, we show that $\sfrac{1}{3}$-approximation on welfare can be achieved by a strategyproof mechanism for binary Leontief utilities (\cref{thm:greedy}). We also give a polytime strategyproof mechanism for binary additive utilities which is \USW{$\sfrac{1}{7}$} in general (\cref{thm:sdadd}) and \USW{$\sfrac{1}{6}$} when agent preferences are symmetric (\cref{cor:sdstarsymm}).

\subsection{Related Work}


Matchings have been studied in a variety of models and with several different objectives. We now give a brief overview of some basics of matching models and defer a comprehensive discussion to \cref{app:relwork}.
%
In the house allocation problem each agent is matched to a house and agents have preferences over the houses. Economic efficiency and strategyproofness has been studied separately \citep{green1979incentives,hylland1979efficient} and together~\cite{krysta2019size}. Unlike our problem, a deterministic serial dictatorship mechanism easily gives better than $2$-approximation for maximum size house allocation. Thus, randomized variations of serial dictatorship have been instrumental in achieving truthful behavior for pareto optimality and envy-freeness \citep{krysta2019size,filos2014social}.
When matching pairs of agents and each side has preference over the other, no strategyproof mechanism exists that produces a stable solution~\citep{roth1982economics}, although finding one stable matching is easy~\citep{gale1962college}. 

The same results hold for
the classical roommate matching problem where pairs of agents are assigned to a rooms, agents have preference over other agents, but are indifferent between rooms \citep{knuth1976mariages, teo1998geometry, aziz2013stable,IRVING1985577}. 
Settings where triples of agents, instead of pairs, need to be matched have also been well studied  \citep{mckay2021three,lam2019existence,BreCheFinNie2020-spscSM-jaamas,DBLP:conf/esa/0001R22}. Maximum welfare and even achieving stability is NP-hard in these settings~\citep{DBLP:conf/esa/0001R22} when each agent has preference over at least a subset of agents (if not all of them). In contrast, in our setting, rooms have no preferences/priorities over the agents. 

Our setting where pairs of agents have to be matched to rooms has been  studied only very recently. Amongst this work, \citet{li2019room, chan2016assignment, huzhang2017online} all consider maximizing social welfare under different settings. None of the work in this space considers any incentives to manipulate. Moreover, they consider settings with only additive utilities.
Leontief utilities have been widely studied by economists and game theorists, particularly in the context of markets and auctions. 
These utilities are useful in capturing complementary preferences. Beyond auctions, they have also been studied in settings like resource allocation \citep{parkes2015beyond,ghodsi2011dominant,brandt2024coordinating} and fair division \citep{branzei2015characterization,bogomolnaia2023guarantees}.
\section{Model}
\label{sec:prelim}
We study a roommate matching model where agents have preferences over  both  their roommate and the room. Our model involves a set  $\N$ of $2n$ agents and a set $\R$ of $n$ rooms. Agent $i\in \N$ has \emph{valuation functions} $v_i$ and $\widehat{v}_i$, where  $v_i:\N \setminus \{i\} \rightarrow \mathbb{R}_{\geq0}$ denotes the ``compatibility values" for the other agents and $\widehat{v}_i: \R \rightarrow \mathbb{R}_{\geq0}$ denotes the values for rooms. Under \emph{binary valuations}, $v_i$ and $\widehat{v}_i$ values are either $0$ or $1$ for each agent $i$. \emph{Symmetric valuations} are where $v_i(j)=v_j(i)$ for each pair of agents $i$ and $j$. 

We denote an instance of roommate matching by the tuple  $I = \langle \N, \R, v,\widehat{v} \rangle$ where $v$ is the vector of compatibility values $(v_1,v_2, \dots, v_{2n})$ and $\widehat{v}$ is the vector of valuations for the rooms $(\widehat{v}_1, \widehat{v}_2, \dots, \widehat{v}_{2n})$. 
We assume that each room is matched to exactly two agents. Formally, for any distinct $i,j \in \N$ and $r\in \R$, we call $(i,j,r)$ a \emph{triple} and we define a roommate matching as follows.

\begin{definition}[Roommate Matching]

A roommate matching $\mu\subseteq \N\times \N \times \R $ is such that each agent or room $a \in \N \cup \R$ is contained in exactly one element (triple) of $\mu$.
\end{definition}

\paragraph{Leontief and Additive Utilities.}
For the ease of presentation, we use  $v_i(\mu)$ and $\widehat{v}_i(\mu)$, respectively, to denote the value of agent $i\in \N$ from their assigned roommate and room under $\mu$, respectively. 
Given a roommate matching $\mu$, we define the \emph{utility} of an agent $i \in \N$ for $\mu$ to be $u_i(\mu) = f(v_i(\mu),\widehat{v}_i(\mu))$ for some computable function $f$ that we call a utility function.
We shall use $u_i(j,r)$ for $i, j \in \N$ and $r\in \R$ to denote the utility of $i$ for being matched to agent $j$ and room $r$. In this paper, the utility function $f$ will be either Leontief or additive\footnote{Some of our results hold for any polynomial-time computable function $f$, as will be stated in the remarks.}. Of the two, the more well-studied one is \textbf{additive utilities} where $u_i(\mu)=\widehat{v}_i(\mu)+v_i(\mu)$. The other is the natural model of \textbf{Leontief utilities} where $u_i(\mu)=\min \{\widehat{v}_i(\mu),v_i(\mu)\}$ (for example, see \citep{parkes2015beyond}). 

\paragraph{Mechanism.} Given an instance  $I=\langle \N, \R, v,\widehat{v} \rangle$, a mechanism $M$ 
uses the valuations $v$ and $\widehat{v}$ reported by the agents and implements a specific algorithm to output a matching $\mu$, written as  $\mu = M(\langle \N, \R, v,\widehat{v} \rangle)$. We distinguish between a mechanism and the algorithm it deploys only when reasoning about agents' incentives to misreport their valuations.
The aim of all the mechanisms we explore is to match agents and rooms in an efficient manner.\footnote{We assume complete matchings with no outside options.} 

\paragraph{Social Welfare.} Given an instance  $I =\langle \N, \R, v,\widehat{v} \rangle$, the \emph{social welfare} (\USW{}) of a matching $\mu$ is defined as the sum of the utilities of all agents, that is $\USW{}(\mu)= \sum_{i\in \N}u_i(\mu)$. A mechanism is social welfare maximizing if it produces a matching with maximum social welfare. Typically, we shall use $\mu^*$ to denote a maximum welfare matching. We say a matching has welfare $\USW{\alpha}$ if its social welfare is at least $\alpha$ fraction of $\USW{}(\mu^*)$ for some $0\leq \alpha \leq 1$. We consider the problem of maximizing \USW{}, both independently and in conjunction with strategyproofness.  

\paragraph{Incentives and Strategyproofness.} We say an agent $i\in \N$ \emph{misreports} their valuation if  their reported valuation $v_i'$ (resp. $\widehat{v}_i'$) is not their true valuation $v_i$ (resp. $\widehat{v}_i)$. For a fixed $i\in \N$, we refer to the set of compatibility valuations of all other agents $\N \setminus \{i\}$ as $v_{-i}$ and the room valuations analogously as $\widehat{v}_{-i}$. Given a mechanism $M$, an agent has an \emph{incentive to misreport}  if their utility $u_i(\mu)$ is strictly less than $u_i(\mu')$ where $\mu$ and $\mu'$ are the output of $M$ for reported valuations $(v,\widehat{v})$  and $(v_i',v_{-i}, \widehat{v}_i', \widehat{v}_{-i})$, respectively. We say that a mechanism is \emph{strategyproof} (SP) if no agent has an incentive to misreport their preferences. 

\begin{definition}[Strategyproofness]
    A roommate matching mechanism $M$ is strategyproof, if for any instance $I=\langle \N,\R, v, \widehat{v}\rangle$ 
    and any agent $i\in \N$, it holds that
    $u_i(\mu)\geq\max_{\widehat{v}_i',v_i'} u_i(\mu')$ 
     for any choice of $v_i'$ and $\widehat{v}'_i$ where $\mu = M(\N, \R, v, \widehat{v})$ and $
     \mu' = M(\N, \R, \{v_i', v_{-i} \}, \{\widehat{v}_i', \widehat{v}_{-i}\})$.
\end{definition}

Our main aim is to develop roommate matching mechanisms that maximize social welfare and are strategyproof under Leontief utilities. We also contrast these results with the case of additive utilities. 
%
%
 When we study SP mechanisms under binary (or symmetric) valuations, we will assume that the {\em mechanism is aware} that valuations are binary (resp. symmetric), so an agent can only misreport with other binary (resp. symmetric) valuation functions. 



\section{Social Welfare Guarantees}\label{sec:swapprox}
 
\noindent We first explore social welfare maximization alone, without considering incentives to manipulate. \citet{chan2016assignment} show that maximizing social welfare under additive utilities is NP-hard and give a \USW{$\sfrac{2}{3}$} algorithm. This algorithm does not extend beyond additive utilities and can result in matchings that are \USW{0} for Leontief utilities (see the example depicted in \cref{fig:ChanCounter} in \cref{app:maxwelfare}.) 
We begin by analyzing the social welfare guarantees provided by maximal matchings.


\subsection{Maximal Matchings}


Maximal matchings provide non-trivial approximations to \USW{} in two-dimensional settings. Typically, such matchings are \USW{$\sfrac{1}{2}$}. Prior work in the roommate matching setting has not considered maximal matchings or even how to generate them. We now discuss different ways of building maximal roommate matchings and the guarantees they provide on \USW{}. 

First, observe that a  maximal matching constructed by arbitrarily choosing \emph{any} triple, irrespective of the agents' preferences could lead to a matching with \USW{0} approximation. 
We shall consider two ways of generating maximal matchings: i) From the preference graph, and ii) from a given set of triples (usually having some specific property, e.g. non-zero \USW{}).

\subsubsection{A Naive Maximal Matching from Preference Graph} Given an instance $I=\langle \N, \R, v, \widehat{v}\rangle$ with {\em binary valuations}, the preference graph $G_I=(X,E)$ is a directed graph with an agent vertex $x_i$ for each agent $i\in \N$ and a room vertex $x_r$ each room $r\in \R$, i.e., $X = \N \cup \R$. There is a directed edge from an agent vertex $x_i$ to a vertex $x\in X$ if agent $i$ has a non-zero value for the agent/room corresponding to vertex $x$ in $I$.

We can generate a \textbf{naive maximal matching} from the preference graph $G_I$ as follows: Pick an edge from $G_I$ and add it to the matching. If both endpoints of the edge are agents, remove both agents' outgoing and incoming edges to/from other agents. Else, if exactly one endpoint is an agent (the other a room), remove all outgoing edges of the agent to any room. If selected edge is incident to an edge selected previously, match the agents and room and remove all incoming and outgoing edges to/from the triple.  Once no edges remain in $G_I$, complete the remaining  matched edges to triples arbitrarily.

\paragraph{Leontief Utilities.} We first consider the naive maximal matching approach for Leontief utilities. While this can be arbitrarily bad for non-symmetric valuations, we provide a simple tweak to get a maximal matching that is \USW{$\sfrac{1}{6}$} under binary symmetric Leontief utilites.

\begin{restatable}{proposition}{maximalunbounded}\label{prop:maximal-unbounded}
    A naive maximal matching cannot achieve better than $\USW{\alpha}$ for any $\alpha > 0$ under Leontief utilities. 
\end{restatable}    

\begin{proof}
    Given a roommate matching instance with binary Leontief utilities $I=\langle \N, \R, v,\widehat{v}\rangle$, let $\mu$ denote a maximal matching obtained by the naive algorithm.
    Observe that it takes two pairs  from the preference graph to match each triple in $\mu$. The choice of one triple can affect at most three optimally matched triples. Formally, let $\mu^*$ be a roommate matching  with maximum $\USW{}$. Let $(i,j,r)$ be a triple matched in the maximal matching $\mu$.  

    Now if the edges selected by the naive algorithm have different outgoing agents  (for example $(i,j)$ and $(j,r)$, as in \cref{fig:maximal-leon}), then both $i$ and $j$ receive utility $0$ under $\mu$.
    That is, if the matched triple is $(i,j,r)$, where $i$ likes $j$ but not $r$ and $j$ likes $r$ but not $i$, both of them get utility $0$ from the triple under $\mu$. 
    Moreover,
    in $\mu^*$, the agents $i$, $j$ and the room $r$ could be matched in three different triples such that each of the six agents matched in the three triples receive utility $1$ under $\mu^*$.
    Further, the remaining four agents (those matched to $i$, $j$ and $r$ under $\mu^*$) may not like any other agent-room pair and also get utility $0$ under $\mu$. We show one such instance in \cref{fig:maximal-leon}. Thus, a maximal matching generated from the preference graph would be \USW{0}.
\end{proof}

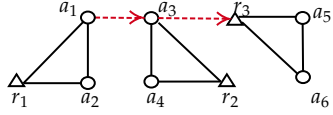
\begin{figure}
    \centering

 \tikzset{every picture/.style={line width=0.75pt}} 

\begin{tikzpicture}[x=0.75pt,y=0.75pt,yscale=-1,xscale=1]

\draw   (13,44) -- (16.83,50) -- (9.18,50) -- cycle ;
\draw   (154.02,45.44) .. controls (154.02,47.1) and (155.58,48.45) .. (157.5,48.45) .. controls (159.42,48.45) and (160.98,47.1) .. (160.98,45.44) .. controls (160.98,43.78) and (159.42,42.43) .. (157.5,42.43) .. controls (155.58,42.43) and (154.02,43.78) .. (154.02,45.44) -- cycle ;
\draw   (153.82,15.58) .. controls (153.82,17.25) and (155.38,18.59) .. (157.3,18.59) .. controls (159.22,18.59) and (160.78,17.25) .. (160.78,15.58) .. controls (160.78,13.92) and (159.22,12.57) .. (157.3,12.57) .. controls (155.38,12.57) and (153.82,13.92) .. (153.82,15.58) -- cycle ;
\draw   (77.5,47.7) .. controls (77.5,49.37) and (79.06,50.71) .. (80.98,50.71) .. controls (82.9,50.71) and (84.46,49.37) .. (84.46,47.7) .. controls (84.46,46.04) and (82.9,44.69) .. (80.98,44.69) .. controls (79.06,44.69) and (77.5,46.04) .. (77.5,47.7) -- cycle ;
\draw   (77.5,15.64) .. controls (77.5,17.31) and (79.06,18.65) .. (80.98,18.65) .. controls (82.9,18.65) and (84.46,17.31) .. (84.46,15.64) .. controls (84.46,13.98) and (82.9,12.63) .. (80.98,12.63) .. controls (79.06,12.63) and (77.5,13.98) .. (77.5,15.64) -- cycle ;
\draw   (123.19,12.99) -- (127.02,19.01) -- (119.36,19.01) -- cycle ;
\draw   (119.07,44) -- (122.9,50) -- (115.25,50) -- cycle ;
\draw    (16.5,48.5) -- (45.5,48.5) ;
\draw    (155,44) -- (126.7,17.7) ;
\draw    (153.82,15) -- (125,15) ;
\draw    (80.98,44.69) -- (80.98,18.65) ;
\draw    (84.6,18) -- (116,47.67) ;
\draw    (157.5,42.43) -- (157.3,18.59) ;
\draw    (16.5,48.5) -- (47,17.2) ;
\draw    (84.46,47.7) -- (116,47.67) ;
\draw [color={rgb, 255:red, 208; green, 2; blue, 27 }  ,draw opacity=1 ] [dash pattern={on 2pt off 1pt}] (52.46,15.38)  --  (77.5,15.64);
\draw [shift={(76,15.36)}, rotate = 180.65] [color={rgb, 255:red, 208; green, 2; blue, 27 }  ,draw opacity=1 ][line width=0.75]    (5.93,-3.29) .. controls (3.95,-1.4) and (1.31,-0.3) .. (0,0) .. controls (1.31,0.3) and (3.95,1.4) .. (5.93,3.29)   ;
\draw [color={rgb, 255:red, 208; green, 2; blue, 27 }  ,draw opacity=1 ] [dash pattern={on 2pt off 1pt}]  (84.46,15.64) -- (119,15.98) ;
\draw [shift={(120,16)}, rotate = 180.56] [color={rgb, 255:red, 208; green, 2; blue, 27 }  ,draw opacity=1 ][line width=0.75]    (5.93,-3.29) .. controls (3.95,-1.4) and (1.31,-0.3) .. (0,0) .. controls (1.31,0.3) and (3.95,1.4) .. (5.93,3.29)   ;
\draw   (45.5,48.42) .. controls (45.5,50.08) and (47.06,51.43) .. (48.98,51.43) .. controls (50.9,51.43) and (52.46,50.08) .. (52.46,48.42) .. controls (52.46,46.76) and (50.9,45.41) .. (48.98,45.41) .. controls (47.06,45.41) and (45.5,46.76) .. (45.5,48.42) -- cycle ;
\draw   (45.5,15.36) .. controls (45.5,17.02) and (47.06,18.37) .. (48.98,18.37) .. controls (50.9,18.37) and (52.46,17.02) .. (52.46,15.36) .. controls (52.46,13.7) and (50.9,12.35) .. (48.98,12.35) .. controls (47.06,12.35) and (45.5,13.7) .. (45.5,15.36) -- cycle ;
\draw    (48.98,45.41) -- (48.98,18.37) ;

\draw (33,6) node [anchor=north west][inner sep=0.75pt]  [font=\scriptsize]  {$a_{1}$};
\draw (44,52) node [anchor=north west][inner sep=0.75pt]  [font=\scriptsize]  {$a_{2}$};
\draw (9,52) node [anchor=north west][inner sep=0.75pt]  [font=\scriptsize]  {$r_{1}$};
\draw (115,52) node [anchor=north west][inner sep=0.75pt]  [font=\scriptsize]  {$r_{2}$};
\draw (121,6) node [anchor=north west][inner sep=0.75pt]  [font=\scriptsize]  {$r_{3}$};
\draw (160,52) node [anchor=north west][inner sep=0.75pt]  [font=\scriptsize]  {$a_{6}$};
\draw (161,9) node [anchor=north west][inner sep=0.75pt]  [font=\scriptsize]  {$a_{5}$};
\draw (77,52) node [anchor=north west][inner sep=0.75pt]  [font=\scriptsize]  {$a_{4}$};
\draw (83,6) node [anchor=north west][inner sep=0.75pt]  [font=\scriptsize]  {$a_{3}$};

\end{tikzpicture}

    \caption{An instance showing that a maximal matching cannot guarantee \USW{$\alpha$} for any $\alpha > 0$  under binary Leontief utilities. The optimal matching is shown in solid black edges and the wrong choice of triples is shown in dashed red edges.}
    \label{fig:maximal-leon}
\end{figure}

Observe that in the example shown in \cref{fig:maximal-leon}, if all agent preferences were symmetric, the naive maximal matching algorithm would be \USW{$\sfrac{1}{6}$}. Unfortunately, this is not true in general.  The naive maximal matching algorithm is \USW{$0$} even for binary symmetric utilities. This is illustrated in the example depicted in \cref{fig:maximal-add}. 

As shown by this example, the problem arises if triples are formed by picking two agent-room edges. We find that if we simply tweak the naive maximal matching to first pick agent-agent edges, we can ensure \USW{$\sfrac{1}{6}$} for binary symmetric Leontief utilities.  

\begin{proposition}
    There is a maximal matching algorithm that is \USW{$\sfrac{1}{6}$} for binary symmetric Leontief utilities .
\end{proposition}

\begin{proof}
    
    Consider a maximal matching algorithm where we first pick an edge between two agents and then an edge between one of these agents and a room.
     This approach \USW{$\sfrac{1}{6}$} for binary symmetric Leontief utilities. 
     
     As in the proof of \cref{prop:maximal-unbounded}, the choice of one triple can disrupt up to three optimally matched triples, where all six agents get utility $1$. The triple selected in $\mu$, would give at least one of the agents a utility of $1$, as the first (agent, agent) edge picked must give both agents value $1$ and the (agent, room) edge picked gives at least one of the two matched agents a value of $1$. If an agent is unmatched, either they had no outgoing edges to other agents in the preference graph, in which case, they get a utility of $0$ under all roommate matchings, or their outgoing edges were disrupted by another agent. Consequently, picking an (agent, agent) edge first for binary symmetric Leontief utilities leads to \USW{$\sfrac{1}{6}$}.    
\end{proof}

\paragraph{Additive utilities.} We now consider additive utilities. We find that the naive maximal matching approach performs much better for binary additive utilities.


    
    %
%

\begin{proposition}
     A naive maximal matching is \USW{$\sfrac{1}{4}$} under additive utilities.
\end{proposition}

\begin{proof}
    Firstly, we show that under additive utilities, a naive maximal matching cannot achieve better than $\USW{\alpha}$ for any $\alpha > \tfrac{1}{4}$. Observe the example with (symmetric) additive utilities shown in \cref{fig:maximal-add}. 
    Here, the optimal matching is denoted with black edges and the dashed red edges represent the (suboptimal) preference edges picked by the naive maximal matching algorithm. Clearly, the optimal solution gives a \USW{} of $8$ while the naive maximal matching has a \USW{} of $2$. Thus, the instance shown is one where the naive maximal matching algorithm is \USW{$\sfrac{1}{4}$}. 
    
    Now we prove the tightness of the bound. Given a roommate matching instance with binary Leontief utilities $I=\langle \N, \R, v,\widehat{v}\rangle$, let $\mu$ be a maximal matching generated by the naive algorithm and $\mu^*$ be a maximum \USW{} matching for $I$.
    
    Let $(i,j,r)\in \mu$.  Observe that unlike the Leontief case, the agents matched to $i$, and $j$ under $\mu^*$ can still be matched to the same room and get utility from it (if any). Similarly, the agents matched to $r$ in $\mu^*$  can still be matched to each other and get the same utility from each other. In order to construct a worst case example, these agents should only get value from $i$, $j$ and $r$. Thus, in order to calculate the approximation bound, it is sufficient to calculate how much welfare $(i,j,r)$ being matched to each other must generate, and how much welfare they each otherwise may have generated in $\mu^*$.
    
    Observe that the two edges selected by the naive algorithm to form $(i,j,r)$ must generate a social welfare of $2$. Each agent being matched suboptimally can reduce the maximum social welfare by at most $3$ ($1$ for their optimal room and $2$ from the compatibility values of their partner under $\mu^*$) and a room being wrongly matched can disrupt a social welfare of up to $2$ ($1$ for each agent matched to it under $\mu^*$). Thus, a complete triple being matched will generate a social welfare of at least $2$ and disrupt an optimal welfare of $8$, ensuring \USW{$\sfrac{1}{4}$}.

    For a triple first matched as a single edge and then completed arbitrarily will have a complete social welfare of at least $1$. Here, we only need to account for those edges not previously disrupted by another triple. If this is an edge between two agents, and even one was matched to a room they like, that was disrupted by another edge matched, and hence, their room values are accounted for. 
    
    If neither were, they could at most have disrupted an optimal welfare of $2$ each from compatibility values, leading to \USW{$\sfrac{1}{4}$}. If this was an agent room edge, at least one other agent who liked this room, could have been matched to it. As it did not, either the welfare is already accounted for, or no such agent exists, and this edge disrupted a social welfare of at most three, leading to \USW{$\sfrac{1}{3}$}. 

    Hence, a maximal matching generated by the naive algorithm must be \USW{$\sfrac{1}{4}$} for binary additive utilities.
    
\end{proof}

\begin{figure}[t]
    \centering

\tikzset{every picture/.style={line width=0.85pt}} 

\begin{tikzpicture}[x=0.75pt,y=0.75pt,yscale=-1,xscale=1]

\draw   (9.14,53.34) -- (13.57,60.31) -- (4.71,60.31) -- cycle ;
\draw   (109.87,60.65) .. controls (109.87,62.58) and (111.67,64.14) .. (113.9,64.14) .. controls (116.12,64.14) and (117.93,62.58) .. (117.93,60.65) .. controls (117.93,58.73) and (116.12,57.17) .. (113.9,57.17) .. controls (111.67,57.17) and (109.87,58.73) .. (109.87,60.65) -- cycle ;
\draw   (149.77,60.82) .. controls (149.77,62.74) and (151.57,64.3) .. (153.79,64.3) .. controls (156.02,64.3) and (157.82,62.74) .. (157.82,60.82) .. controls (157.82,58.89) and (156.02,57.33) .. (153.79,57.33) .. controls (151.57,57.33) and (149.77,58.89) .. (149.77,60.82) -- cycle ;
\draw   (202.45,59.8) .. controls (202.45,61.72) and (204.25,63.28) .. (206.47,63.28) .. controls (208.7,63.28) and (210.5,61.72) .. (210.5,59.8) .. controls (210.5,57.88) and (208.7,56.32) .. (206.47,56.32) .. controls (204.25,56.32) and (202.45,57.88) .. (202.45,59.8) -- cycle ;
\draw   (202.45,22.69) .. controls (202.45,24.61) and (204.25,26.17) .. (206.47,26.17) .. controls (208.7,26.17) and (210.5,24.61) .. (210.5,22.69) .. controls (210.5,20.77) and (208.7,19.21) .. (206.47,19.21) .. controls (204.25,19.21) and (202.45,20.77) .. (202.45,22.69) -- cycle ;
\draw   (113.54,18.46) -- (117.97,25.43) -- (109.11,25.43) -- cycle ;
\draw   (250.57,52.8) -- (255,59.76) -- (246.14,59.76) -- cycle ;
\draw    (113.9,57.17) -- (113.2,25.2) ;
\draw    (153.79,57.33) -- (117.97,25.43) ;
\draw    (206.47,56.32) -- (206.47,26.17) ;
\draw    (209,25) -- (249,56) ;
\draw    (11.92,56.9) -- (48.5,24.49) ;
\draw [color={rgb, 255:red, 208; green, 2; blue, 27 }  ,draw opacity=1 ] [dash pattern={on 2pt off 1pt}]   (54.82,22.36) -- (109.27,22.49) ;
\draw [color={rgb, 255:red, 208; green, 2; blue, 27 }  ,draw opacity=1 ] [dash pattern={on 2pt off 1pt}]   (202.45,22.69) -- (118.29,22.5) ;
\draw   (46.76,60.63) .. controls (46.76,62.55) and (48.57,64.11) .. (50.79,64.11) .. controls (53.01,64.11) and (54.82,62.55) .. (54.82,60.63) .. controls (54.82,58.7) and (53.01,57.14) .. (50.79,57.14) .. controls (48.57,57.14) and (46.76,58.7) .. (46.76,60.63) -- cycle ;
\draw   (46.76,22.36) .. controls (46.76,24.28) and (48.57,25.84) .. (50.79,25.84) .. controls (53.01,25.84) and (54.82,24.28) .. (54.82,22.36) .. controls (54.82,20.43) and (53.01,18.87) .. (50.79,18.87) .. controls (48.57,18.87) and (46.76,20.43) .. (46.76,22.36) -- cycle ;
\draw    (50.79,57.14) -- (50.79,25.84) ;

\draw (40,8) node [anchor=north west][inner sep=0.75pt]  [font=\scriptsize]  {$a_{1}$};
\draw (44,66) node [anchor=north west][inner sep=0.75pt]  [font=\scriptsize]  {$a_{2}$};
\draw (5,66) node [anchor=north west][inner sep=0.75pt]  [font=\scriptsize]  {$r_{1}$};
\draw (246,66) node [anchor=north west][inner sep=0.75pt]  [font=\scriptsize]  {$r_{2}$};
\draw (108.66,8) node [anchor=north west][inner sep=0.75pt]  [font=\scriptsize]  {$r_{3}$};
\draw (150,66) node [anchor=north west][inner sep=0.75pt]  [font=\scriptsize]  {$a_{6}$};
\draw (108.65,66) node [anchor=north west][inner sep=0.75pt]  [font=\scriptsize]  {$a_{5}$};
\draw (202,66) node [anchor=north west][inner sep=0.75pt]  [font=\scriptsize]  {$a_{4}$};
\draw (203,10) node [anchor=north west][inner sep=0.75pt]  [font=\scriptsize]  {$a_{3}$};

\end{tikzpicture}

    \caption{An instance showing that a maximal matching cannot guarantee \USW{$\alpha$} for any $\alpha > \sfrac{1}{4}$ under binary additive utilities. The optimal matching is shown in solid black edges and the wrong choice of triple is shown in dashed red edges.}
    \label{fig:maximal-add}
\end{figure}

\subsubsection{Matching with Triangles and L-Shapes}

\cref{prop:maximal-unbounded} shows that a naive maximal matching could provide no welfare guarantee under Leontief utilities. It illustrates the challenges of adopting graph-matching algorithms to roommate matching instances, especially under Leontief.

We seek to develop a polynomial-time algorithm that achieves a constant approximation of welfare. To this end, it is necessary to understand the structure of roommate matching problems. Thus, we start by introducing triple structures according to their contribution to welfare.

\paragraph{Triple Structures.}
Given an instance of the roommate matching problem, a triple $(i,j,r)$ can be represented as a graphical structure of the preference graph. These structures are illustrated graphically in \cref{fig:structures}. 
A triple may form a triangle match (T) if (i) both agents like each other and (ii) both like the matched room. 
If only one edge between an agent and a room is missing in a triangle match, that is both agents like each other, but only one likes the room, then we say the match is an L-shaped (L) match.

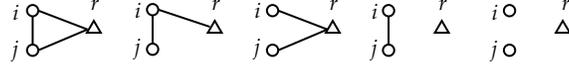
\begin{figure}[t]
\centering

\tikzset{every picture/.style={line width=0.75pt}} 

\begin{tikzpicture}[x=0.75pt,y=0.75pt,yscale=-1,xscale=1]

\draw   (80.8,44.85) .. controls (80.8,43.28) and (82.08,42) .. (83.65,42) .. controls (85.22,42) and (86.5,43.28) .. (86.5,44.85) .. controls (86.5,46.42) and (85.22,47.7) .. (83.65,47.7) .. controls (82.08,47.7) and (80.8,46.42) .. (80.8,44.85) -- cycle ;
\draw   (80.8,64.85) .. controls (80.8,63.28) and (82.08,62) .. (83.65,62) .. controls (85.22,62) and (86.5,63.28) .. (86.5,64.85) .. controls (86.5,66.42) and (85.22,67.7) .. (83.65,67.7) .. controls (82.08,67.7) and (80.8,66.42) .. (80.8,64.85) -- cycle ;
\draw   (111,56) -- (118,56) -- (114.5,50) -- cycle ;
\draw    (83.65,47.7) -- (83.65,62) ;
\draw    (86.5,44.85) -- (111.3,54.35) ;
\draw    (86.5,64.85) -- (111.3,54.35) ;

\draw   (141.3,44.35) .. controls (141.3,42.78) and (142.58,41.5) .. (144.15,41.5) .. controls (145.72,41.5) and (147,42.78) .. (147,44.35) .. controls (147,45.92) and (145.72,47.2) .. (144.15,47.2) .. controls (142.58,47.2) and (141.3,45.92) .. (141.3,44.35) -- cycle ;
\draw   (141.3,64.35) .. controls (141.3,62.78) and (142.58,61.5) .. (144.15,61.5) .. controls (145.72,61.5) and (147,62.78) .. (147,64.35) .. controls (147,65.92) and (145.72,67.2) .. (144.15,67.2) .. controls (142.58,67.2) and (141.3,65.92) .. (141.3,64.35) -- cycle ;
\draw   (175.5,50) -- (172,56) -- (179,56) -- cycle ;
\draw    (144.15,47.2) -- (144.15,61.5) ;
\draw    (147,44.35) -- (173,53) ;
\draw   (201.8,44.85) .. controls (201.8,43.28) and (203.08,42) .. (204.65,42) .. controls (206.22,42) and (207.5,43.28) .. (207.5,44.85) .. controls (207.5,46.42) and (206.22,47.7) .. (204.65,47.7) .. controls (203.08,47.7) and (201.8,46.42) .. (201.8,44.85) -- cycle ;
\draw   (201.8,64.85) .. controls (201.8,63.28) and (203.08,62) .. (204.65,62) .. controls (206.22,62) and (207.5,63.28) .. (207.5,64.85) .. controls (207.5,66.42) and (206.22,67.7) .. (204.65,67.7) .. controls (203.08,67.7) and (201.8,66.42) .. (201.8,64.85) -- cycle ;
\draw   (231.5,56) -- (238.5,56) -- (235,50) -- cycle ;
\draw    (207.5,44.85) -- (232.3,54.35) ;
\draw    (207.5,64.85) -- (232.3,54.35) ;
\draw   (260.3,44.85) .. controls (260.3,43.28) and (261.58,42) .. (263.15,42) .. controls (264.72,42) and (266,43.28) .. (266,44.85) .. controls (266,46.42) and (264.72,47.7) .. (263.15,47.7) .. controls (261.58,47.7) and (260.3,46.42) .. (260.3,44.85) -- cycle ;
\draw   (260.3,64.85) .. controls (260.3,63.28) and (261.58,62) .. (263.15,62) .. controls (264.72,62) and (266,63.28) .. (266,64.85) .. controls (266,66.42) and (264.72,67.7) .. (263.15,67.7) .. controls (261.58,67.7) and (260.3,66.42) .. (260.3,64.85) -- cycle ;
\draw   (286.5,56) -- (293.5,56) -- (290,50) -- cycle ;
\draw    (263.15,47.7) -- (263.15,62) ;
\draw   (321.3,44.85) .. controls (321.3,43.28) and (322.58,42) .. (324.15,42) .. controls (325.72,42) and (327,43.28) .. (327,44.85) .. controls (327,46.42) and (325.72,47.7) .. (324.15,47.7) .. controls (322.58,47.7) and (321.3,46.42) .. (321.3,44.85) -- cycle ;
\draw   (321.3,64.85) .. controls (321.3,63.28) and (322.58,62) .. (324.15,62) .. controls (325.72,62) and (327,63.28) .. (327,64.85) .. controls (327,66.42) and (325.72,67.7) .. (324.15,67.7) .. controls (322.58,67.7) and (321.3,66.42) .. (321.3,64.85) -- cycle ;
\draw   (347.5,56) -- (354.59,56) -- (351,50) -- cycle ;

\draw (72,40.4) node [anchor=north west][inner sep=0.75pt]  [font=\scriptsize]  {$i$};
\draw (71.5,58.4) node [anchor=north west][inner sep=0.75pt]  [font=\scriptsize]  {$j$};
\draw (111.5,38) node [anchor=north west][inner sep=0.75pt]  [font=\scriptsize]  {$r$};
\draw (133.5,39.4) node [anchor=north west][inner sep=0.75pt]  [font=\scriptsize]  {$i$};
\draw (133,57.4) node [anchor=north west][inner sep=0.75pt]  [font=\scriptsize]  {$j$};
\draw (194,39.4) node [anchor=north west][inner sep=0.75pt]  [font=\scriptsize]  {$i$};
\draw (193.5,57.4) node [anchor=north west][inner sep=0.75pt]  [font=\scriptsize]  {$j$};
\draw (253,39.9) node [anchor=north west][inner sep=0.75pt]  [font=\scriptsize]  {$i$};
\draw (252.5,57.9) node [anchor=north west][inner sep=0.75pt]  [font=\scriptsize]  {$j$};
\draw (313.5,40.4) node [anchor=north west][inner sep=0.75pt]  [font=\scriptsize]  {$i$};
\draw (313,58.4) node [anchor=north west][inner sep=0.75pt]  [font=\scriptsize]  {$j$};
\draw (172.5,38) node [anchor=north west][inner sep=0.75pt]  [font=\scriptsize]  {$r$};
\draw (232,38) node [anchor=north west][inner sep=0.75pt]  [font=\scriptsize]  {$r$};
\draw (287,38) node [anchor=north west][inner sep=0.75pt]  [font=\scriptsize]  {$r$};
\draw (349,38) node [anchor=north west][inner sep=0.75pt]  [font=\scriptsize]  {$r$};

\end{tikzpicture}

    \caption{Possible matching structures in a roommate matching problem. The first structure on the left is a triangle (T), the next one is L-shaped (L) and  gives non-zero utility to agent $i$. The last three show the cases where both $i$ and $j$ get utility $0$ under Leontief.  }
    \label{fig:structures}
\end{figure}

\begin{definition}
    A matching is L/T-matching if it only consists of L-shapes and triangles.
    An L/T-matching $M$ is maximal if it admits no additional L/T triples.
\end{definition}

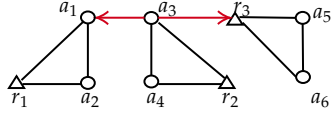
\begin{figure}[t]
    \centering

\tikzset{every picture/.style={line width=0.75pt}} 

\begin{tikzpicture}[x=0.75pt,y=0.75pt,yscale=-1,xscale=1]

\draw   (13,45) -- (16.83,51) -- (9.18,51) -- cycle ;
\draw   (154.02,45.44) .. controls (154.02,47.1) and (155.58,48.45) .. (157.5,48.45) .. controls (159.42,48.45) and (160.98,47.1) .. (160.98,45.44) .. controls (160.98,43.78) and (159.42,42.43) .. (157.5,42.43) .. controls (155.58,42.43) and (154.02,43.78) .. (154.02,45.44) -- cycle ;
\draw   (153.82,15.58) .. controls (153.82,17.25) and (155.38,18.59) .. (157.3,18.59) .. controls (159.22,18.59) and (160.78,17.25) .. (160.78,15.58) .. controls (160.78,13.92) and (159.22,12.57) .. (157.3,12.57) .. controls (155.38,12.57) and (153.82,13.92) .. (153.82,15.58) -- cycle ;
\draw   (77.5,47.7) .. controls (77.5,49.37) and (79.06,50.71) .. (80.98,50.71) .. controls (82.9,50.71) and (84.46,49.37) .. (84.46,47.7) .. controls (84.46,46.04) and (82.9,44.69) .. (80.98,44.69) .. controls (79.06,44.69) and (77.5,46.04) .. (77.5,47.7) -- cycle ;
\draw   (77.5,15.64) .. controls (77.5,17.31) and (79.06,18.65) .. (80.98,18.65) .. controls (82.9,18.65) and (84.46,17.31) .. (84.46,15.64) .. controls (84.46,13.98) and (82.9,12.63) .. (80.98,12.63) .. controls (79.06,12.63) and (77.5,13.98) .. (77.5,15.64) -- cycle ;
\draw   (123.19,12.99) -- (127.02,19.01) -- (119.36,19.01) -- cycle ;
\draw   (119.07,45) -- (122.9,51) -- (115.25,51) -- cycle ;
\draw    (16,49) -- (45.5,49) ;
\draw    (154.02,45.44) -- (127.02,18.01) ;
\draw    (153.82,15.58) -- (125.4,15.2) ;
\draw    (80.98,44.69) -- (80.98,18.65) ;
\draw    (84.6,18) -- (118,46) ;
\draw    (157.5,42.43) -- (157.3,18.59) ;
\draw    (15.5,47) -- (47,17.2) ;
\draw    (84.46,49) -- (117,49) ;
\draw [color={rgb, 255:red, 208; green, 2; blue, 27 }  ,draw opacity=1 ]   (77.5,15.5) -- (54.46,15.5) ;
\draw [shift={(53.3,15.5)}, rotate = 0.65] [color={rgb, 255:red, 208; green, 2; blue, 27 }  ,draw opacity=1 ][line width=0.75]    (5.93,-3.29) .. controls (3.95,-1.4) and (1.31,-0.3) .. (0,0) .. controls (1.31,0.3) and (3.95,1.4) .. (5.93,3.29)   ;
\draw [color={rgb, 255:red, 208; green, 2; blue, 27 }  ,draw opacity=1 ]   (84.46,15.5) -- (119,15.5) ;
\draw [shift={(120.5,15.5)}, rotate = 180.56] [color={rgb, 255:red, 208; green, 2; blue, 27 }  ,draw opacity=1 ][line width=0.75]    (5.93,-3.29) .. controls (3.95,-1.4) and (1.31,-0.3) .. (0,0) .. controls (1.31,0.3) and (3.95,1.4) .. (5.93,3.29)   ;
\draw   (45.5,48.42) .. controls (45.5,50.08) and (47.06,51.43) .. (48.98,51.43) .. controls (50.9,51.43) and (52.46,50.08) .. (52.46,48.42) .. controls (52.46,46.76) and (50.9,45.41) .. (48.98,45.41) .. controls (47.06,45.41) and (45.5,46.76) .. (45.5,48.42) -- cycle ;
\draw   (45.5,15.36) .. controls (45.5,17.02) and (47.06,18.37) .. (48.98,18.37) .. controls (50.9,18.37) and (52.46,17.02) .. (52.46,15.36) .. controls (52.46,13.7) and (50.9,12.35) .. (48.98,12.35) .. controls (47.06,12.35) and (45.5,13.7) .. (45.5,15.36) -- cycle ;
\draw    (48.98,45.41) -- (48.98,18.37) ;

\draw (33,6) node [anchor=north west][inner sep=0.75pt]  [font=\scriptsize]  {$a_{1}$};
\draw (44,52) node [anchor=north west][inner sep=0.75pt]  [font=\scriptsize]  {$a_{2}$};
\draw (9,52) node [anchor=north west][inner sep=0.75pt]  [font=\scriptsize]  {$r_{1}$};
\draw (115,52) node [anchor=north west][inner sep=0.75pt]  [font=\scriptsize]  {$r_{2}$};
\draw (121,6) node [anchor=north west][inner sep=0.75pt]  [font=\scriptsize]  {$r_{3}$};
\draw (160,52) node [anchor=north west][inner sep=0.75pt]  [font=\scriptsize]  {$a_{6}$};
\draw (161,9) node [anchor=north west][inner sep=0.75pt]  [font=\scriptsize]  {$a_{5}$};
\draw (77,52) node [anchor=north west][inner sep=0.75pt]  [font=\scriptsize]  {$a_{4}$};
\draw (83,6) node [anchor=north west][inner sep=0.75pt]  [font=\scriptsize]  {$a_{3}$};

\end{tikzpicture}

    \caption{Tight Example of \USW{$\sfrac{1}{6}$} for \cref{thm:LT-maximal}}
    \label{fig:LTmaximal}
\end{figure}

\begin{theorem}\label{thm:LT-maximal}
    An L/T-maximal matching guarantees $\USW{\sfrac{1}{6}}$ under binary Leontief utilities.
\end{theorem}

\begin{proof}

Given a roommate matching instance $I=\langle \N, \R, v, \widehat{v}\rangle$ with binary Leontief utilities, let $\mu^*$ be a maximum \USW{} roommate matching. If $I$ has no Ls or Ts then clearly, every agent gets a utility of $0$ from every roommate matching. Thus, in this case, the approximation bound holds. 

If on the other hand, $\mu^*$ does indeed contain triples that are Ls or Ts, then the wrong choice of $L$ can intersect at most three  matched triples in $\mu^*$, which could in the worst case all be Ts. As a result, the maximal matching would only give one agent a utility of $1$, whereas $\mu^*$ gives a utility of $1$ to all six agents. Thus, an L/T-maximal matching is \USW{$\sfrac{1}{6}$} under binary Leontief utilities. 
\end{proof}

A worst case example for \cref{thm:LT-maximal} is shown in \cref{fig:LTmaximal}. Here an optimal solution would match the three triangles, but picking the L formed by $(a_1,a_3,r_3)$ is maximal. As a result, an L/T-maximal matching is \USW{$\sfrac{1}{6}$} under binary Leontief utilities.

\begin{remark}
    For binary additive utilities, clearly, the types of triples would be more diverse and hence, building a maximal matching from triples with non-zero utility could lead to $\USW{\sfrac{1}{12}}$ in the worst case.
\end{remark}

\subsubsection{Triangle-then-L Maximal Matching}
Our previous analysis shows that it is important to not eliminate triples with more social welfare too early. Thus, we now propose an algorithm based on repeatedly choosing the highest utility triple (triangles first then Ls) produces a maximal matching. This approach guarantees \USW{$\sfrac{1}{3}$}, for unrestricted valuations and all utility types. 
\begin{algorithm}[t]
  \KwIn{ Roommate Matching instance $\langle \N,\R, v,\widehat{v} \rangle$ }
  \KwOut{A matching $\mu$}
  Initialize matching $\mu \gets \emptyset$ \;
  \While{ there exists an unmatched triple in $\N \cup \R$}{
  
    Let the set of unmatched triples  $P = \{(i,j,r) | i,j\in \N$ and $r\in \R \} $\; 
    $(i,j,r)\gets \argmax _{(i,j,r)\in P}(u_i(\{j,r\})+u_j(\{i,r\}))$ \Comment{Tie-breaking is arbitrary}\;
    Update $\mu \gets \mu \cup (i,j,r)$\;
    Update $\N \gets \N \setminus \{i,j\}$; $\R \gets \R \setminus \{r\}$\;
    }
   \caption{Triangle-then-L Maximal Matching}\label{alg:greedy}
\end{algorithm}

\begin{restatable}{theorem}{greedyalgo}\label{thm:greedyalg}
    A Triangle-then-L maximal matching is \USW{$\sfrac{1}{3}$}.
\end{restatable}

\begin{proof}
    \begin{sloppypar}
Fix an arbitrary roommate matching instance $I=\langle \N,\R, v,\widehat{v} \rangle$. Let $\mu$ be a Triangle-then-L maximal matching (returned by Algorithm \ref{alg:greedy}) on $I$. Let $\mu^*$ be a social welfare maximizing roommate matching on $I$. We show that this algorithm always returns a \USW{$\sfrac{1}{3}$} roommate matching. 
\end{sloppypar}

Let $(i,j,r)$ be the first triple to be matched in \cref{alg:greedy} when run on $I$. Under $\mu^*$ all three of $i$, $j$ and $r$ could be matched very differently from $\mu$. Their partners under $\mu^*$ may get utility $0$ whenever they are not matched as in $\mu^*$. Let $i$, $j$, and $r$ be matched in $\mu^*$ as $\{(i,i',r_i),(j,j',r_j), (i_r,j_r,r)\}\subseteq \mu^*$.   

In the worst case, all three of these triples have the same social welfare as $(i,j,r)$, and these remaining agents have no value when matched in any other triple.  That is, $u_i(\{j,r\})+u_j(\{i,r\})=u_i(\{i',r_i\})+u_{i'}(\{i,r_i\})=u_j(\{j',r_j\})+u_{j'}(\{j,r_j\})=u_{i_r}(\{j_r,r\})+u_{j_r}(\{i_r,r\})$. Consequently, by picking $(i,j,r)$, the greedy mechanism achieves only $\sfrac{1}{3}$ of the optimal $\USW{}$ achieved by agents $i, i', j, j',i_r$ and $j_r$ in $\mu^*$. 
Similarly, for any triple matched in the greedy mechanism, there may be three incident triples of equal value in $\mu^*$, a social welfare maximizing matching. Thus, social welfare of $\mu$ is at least $\sfrac{1}{3}$ times the social welfare of $\mu^*$.
\end{proof}

\begin{remark}
    The approximation guarantee of the Triangle-then-L mechanism (\cref{alg:greedy}) does not depend on how the agents combine the values they receive from their room and partner. Consequently, this result is applicable to any computable utility function  (e.g., submodular or monotonic).
\end{remark}


\subsection{Computational Complexity}
We now show that maximizing social welfare is APX-hard, even under binary symmetric Leontief utilities. 
We present a factor preserving reduction from $3$SAT. The construction will be similar to the standard reduction from $3$SAT to \textsc{$3$-Dimensional Matching} \citep{GJ79}. Interestingly, we show the problem is hard even if there are only Ls. We defer the full proof to \cref{app:red}.

\begin{restatable}{theorem}{swhardl}\label{thm:SW-hardL}
    Finding a maximum \USW{} roommate matching is APX-hard under binary Leontief utilities. 
\end{restatable}

In particular, the instance creates a setting where all triples have utility at most $1$. That is a maximum welfare matching would essentially be a maximal matching of Ls with maximum size. This is an important difference from the reduction in \cite{chan2016assignment} which proves NP-hardness with mixtures of Ls and Ts. 

In fact, our reduction shows that the problem is NP-hard even when each agent has degree three in the preference graph, i.e., the problem in paraNP-hard\footnote{We cannot expect to get an algorithm that runs in $n^d$ where $d$ is the degree of the agents.} with respect to degree of agents. 
\begin{corollary}
Maximizing \USW{} is NP-hard under Leontief utilities in the marriage setting even when agents' valuations are binary, symmetric, and each agent likes at most one room and at most two other agents.
\end{corollary}

In light of this result, we now show an approximation algorithm via the {\sc Set Packing} problem irrespective of agents' underlying utility function.

\subsubsection{Welfare Approximation via \textsc{$3$-Set Packing}}
\begin{sloppypar}
In \textsc{Weighted $3$-Set Packing} we are given a universe $\mathcal{U}$ and a collection $\mathcal{C}$ of weighted $3$-sized sets over $\mathcal{U}$, and the goal is to obtain a subset of $\mathcal{C}$ of pairwise disjoint sets with maximum total weight.    
It is well-known that \textsc{$3$-Set Packing} generalises $3$-dimensional matching. We  show that our problem reduces to \textsc{$3$-Set Packing} and can utilize the prior work on it to obtain approximation and exact algorithms. 
\end{sloppypar} 

\begin{restatable}{proposition}{setpackbound}\label{prop:setpackbound}
There exists an approximation preserving reduction  from  welfare maximization in roommate matchings  to \textsc{Weighted $3$-Set Packing}. This gives an algorithm which is \USW{$0.559$} under Leontief utilities.
\end{restatable}

\begin{proof}

Given a roommate matching instance $\langle \N, \R, v, \widehat{v}\rangle$, we construct an instance of \textsc{Weighted $3$-Set Packing} where the universe $\mathcal{U}$ and the collection of subsets $\mathcal{C}$ are defined as follows.

All the agents and the rooms are the elements of the universe $\mathcal{U}$.
For each triple consisting of a pair of (distinct) agents and a room, we create a set in $\mathcal{C}$ and we assign the weight for this set to be the sum of the utilities of the two agents in the set. The utility of an agent in a triple is computed based on the assumed utility function (Leontief or additive). Observe that this does not depend on whether the utilities are Leontief or additive.  Clearly, a set packing of weight $k$ exists in the constructed instance if and only if a roommate matching of \USW{} $k$ exists in the original instance. %

Consequently, the approximation algorithm in~\cite{thiery2023improved} gives a $0.559$-approximation for our problem. 
\end{proof}

\noindent We can obtain an exact algorithm via the known exact algorithm for \textsc{Weighted $3$-Set Packing}~\citep{zehaviforgetfulness}.

\begin{corollary}\label{cor:setpackboundFPT}
There exists an algorithm to maximize the \USW{} under Leontief utilities for arbitrary valuations that runs 
in time $O^*(8.097^n)$.
\end{corollary}



\section{Strategyproof Mechanisms}
We now show upper and lower bounds on the approximation guarantees possible under strategyproof mechanisms under both Leontief and additive utilities.


\subsection{Impossibility Results}

We start by finding that strategyproofness is often incompatible with welfare under some natural utility models (e.g. general valuations, additive utilities). These impossibilities motivate the study of strategyproof mechanism for restricted valuations (e.g. binary valuations).

\begin{restatable}{theorem}{spimpossible}\label{thm:no-alpha-approx}
There exists no strategyproof mechanism that can guarantee \USW{$\alpha$} for any $0<\alpha \leq 1$, for general (unrestricted) valuations under either Leontief or additive utilities.
\end{restatable}

\begin{proof}
We shall establish this result by giving an instance for each value of $\alpha \in (0,1]$ where for every matching that is \USW{$\alpha$}, there is an agent who has an incentive to misreport their preferences. 

Fix a value of $\alpha \in (0,1]$. Consider an instance with 4 agents and 2 rooms, where each agent $i\in \N$ has value $\widehat{v}_i(r)=k$ for every room $r\in \R $, for $k\geq 0$. We will fix this value later (depending on whether the utilities are additive or Leontief), but it is important that each agent is indifferent among the rooms. Consequently, the agents can be strategic about their compatibility values only, if at all. 

The agent compatibility values are depicted in Figure \ref{fig:noapproxSW}. There are four agents $\{a_1,a_2,a_3,a_4\}$. Agent $a_1$ has value $1$ for $a_2$, agent $a_2$ has value $1$ for $a_3$, agent $a_3$ has value $1$ for $a_4$ and agent $a_4$ has value $1$ for $a_1$. That is, $v_{a_1}(a_2)=v_{a_2}(a_3)=v_{a_3}(a_4)=v_{a_4}(a_1)=1$. All other agent values are $0$. The two possible types of maximum welfare matchings are depicted by different color edges. 

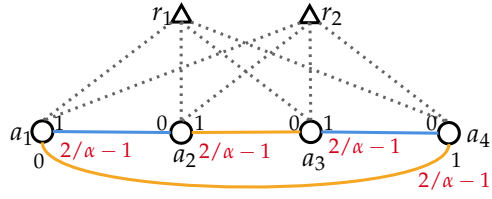
\begin{figure}[t]
    \centering

\tikzset{every picture/.style={line width=1.15pt}} 

\begin{tikzpicture}[x=0.75pt,y=0.75pt,yscale=-1,xscale=1]

\draw   (221.77,70.6) .. controls (221.77,67.62) and (224.13,65.21) .. (227.05,65.21) .. controls (229.96,65.21) and (232.32,67.62) .. (232.32,70.6) .. controls (232.32,73.58) and (229.96,76) .. (227.05,76) .. controls (224.13,76) and (221.77,73.58) .. (221.77,70.6) -- cycle ;
\draw   (152.1,69.94) .. controls (152.1,66.96) and (154.47,64.54) .. (157.38,64.54) .. controls (160.29,64.54) and (162.65,66.96) .. (162.65,69.94) .. controls (162.65,72.92) and (160.29,75.33) .. (157.38,75.33) .. controls (154.47,75.33) and (152.1,72.92) .. (152.1,69.94) -- cycle ;
\draw   (86.77,70.27) .. controls (86.77,67.29) and (89.13,64.88) .. (92.05,64.88) .. controls (94.96,64.88) and (97.32,67.29) .. (97.32,70.27) .. controls (97.32,73.25) and (94.96,75.67) .. (92.05,75.67) .. controls (89.13,75.67) and (86.77,73.25) .. (86.77,70.27) -- cycle ;
\draw   (16.77,69.77) .. controls (16.77,66.79) and (19.13,64.38) .. (22.05,64.38) .. controls (24.96,64.38) and (27.32,66.79) .. (27.32,69.77) .. controls (27.32,72.75) and (24.96,75.17) .. (22.05,75.17) .. controls (19.13,75.17) and (16.77,72.75) .. (16.77,69.77) -- cycle ;
\draw [color={rgb, 255:red, 74; green, 144; blue, 226 }  ,draw opacity=1 ]   (27.32,69.77) -- (86.77,70.27) ;
\draw [color={rgb, 255:red, 74; green, 144; blue, 226 }  ,draw opacity=1 ]   (162.65,69.94) -- (221.77,70.6) ;
\draw [color={rgb, 255:red, 245; green, 166; blue, 35 }  ,draw opacity=1 ]   (152.1,69.94) -- (97.32,70.27) ;
\draw [color={rgb, 255:red, 245; green, 166; blue, 35 }  ,draw opacity=1 ]   (22.05,75.17) .. controls (22,104.83) and (227,105.5) .. (227.05,76) ;
\draw   (92.17,6) -- (97,14.5) -- (87.35,14.5) -- cycle ;
\draw   (156.84,6.33) -- (161.67,14.83) -- (152.01,14.83) -- cycle ;
\draw [color={rgb, 255:red, 100; green, 100; blue, 100 }  ,draw opacity=1 ] [dash pattern={on 1pt off 2pt}]   (87.35,14.5) -- (22.05,64.38) ;
\draw [color={rgb, 255:red, 100; green, 100; blue, 100 }  ,draw opacity=1 ] [dash pattern={on 1pt off 2pt}]   (152.01,14.83) -- (22.05,64.38) ;
\draw [color={rgb, 255:red, 100; green, 100; blue, 100 }  ,draw opacity=1 ] [dash pattern={on 1pt off 2pt}]   (161.67,14.83) -- (227.05,65.21) ;
\draw [color={rgb, 255:red, 100; green, 100; blue, 100 }  ,draw opacity=1 ]  [dash pattern={on 1pt off 2pt}]  (157,14.83) -- (157.38,64.54) ;
\draw [color={rgb, 255:red, 100; green, 100; blue, 100 }  ,draw opacity=1 ] [dash pattern={on 1pt off 2pt}]   (157,14.83) -- (92.05,64.88) ;
\draw [color={rgb, 255:red, 100; green, 100; blue, 100 }  ,draw opacity=1 ] [dash pattern={on 1pt off 2pt}]   (97,14.5) -- (227.05,65.21) ;
\draw [color={rgb, 255:red, 100; green, 100; blue, 100 }  ,draw opacity=1 ] [dash pattern={on 1pt off 2pt}]   (92,14.5) -- (157.38,64.54) ;
\draw [color={rgb, 255:red, 100; green, 100; blue, 100 }  ,draw opacity=1 ] [dash pattern={on 1pt off 2pt}]   (92,14.5) -- (92.05,64.88) ;

\draw (26.26,60) node [anchor=north west][inner sep=0.75pt]  [font=\scriptsize]  {$1$};
\draw (78,60) node [anchor=north west][inner sep=0.75pt]  [font=\scriptsize]  {$0$};
\draw (97,60) node [anchor=north west][inner sep=0.75pt]  [font=\scriptsize]  {$1$};
\draw (213.9,60) node [anchor=north west][inner sep=0.75pt]  [font=\scriptsize]  {$0$};
\draw (16,79.07) node [anchor=north west][inner sep=0.75pt]  [font=\scriptsize]  {$0$};
\draw (144.68,60) node [anchor=north west][inner sep=0.75pt]  [font=\scriptsize]  {$0$};
\draw (226.34,78.59) node [anchor=north west][inner sep=0.75pt]  [font=\scriptsize]  {$1$};
\draw (162,60) node [anchor=north west][inner sep=0.75pt]  [font=\scriptsize]  {$1$};
\draw (210,88) node [anchor=north west][inner sep=0.75pt]  [font=\scriptsize,color={rgb, 255:red, 208; green, 2; blue, 27 }  ,opacity=1 ]  {$2/\alpha -1$};
\draw (99.32,73.67) node [anchor=north west][inner sep=0.75pt]  [font=\scriptsize,color={rgb, 255:red, 208; green, 2; blue, 27 }  ,opacity=1 ]  {$2/\alpha -1$};
\draw (29.32,73.17) node [anchor=north west][inner sep=0.75pt]  [font=\scriptsize,color={rgb, 255:red, 208; green, 2; blue, 27 }  ,opacity=1 ]  {$2/\alpha -1$};
\draw (164.68,71.02) node [anchor=north west][inner sep=0.75pt]  [font=\scriptsize,color={rgb, 255:red, 208; green, 2; blue, 27 }  ,opacity=1 ]  {$2/\alpha -1$};
\draw (4,65.78) node [anchor=north west][inner sep=0.75pt]  [font=\small]  {$a_{1}$};
\draw (234.38,66.59) node [anchor=north west][inner sep=0.75pt]  [font=\small]  {$a_{4}$};
\draw (86.24,76.97) node [anchor=north west][inner sep=0.75pt]  [font=\small]  {$a_{2}$};
\draw (151.33,78.78) node [anchor=north west][inner sep=0.75pt]  [font=\small]  {$a_{3}$};
\draw (75.62,5.07) node [anchor=north west][inner sep=0.75pt]  [font=\small]  {$r_{1}$};
\draw (161,5.4) node [anchor=north west][inner sep=0.75pt]  [font=\small]  {$r_{2}$};

\end{tikzpicture}

    \caption{Agent compatibility values and misreports (in red) for the instance described in \cref{thm:no-alpha-approx}. The different colors on the edge correspond to the different types of maximum \USW{} matchings for this instance. }
    \label{fig:noapproxSW}
\end{figure}

\paragraph{Additive Utilities.} We first prove the theorem when the utilities are additive. We set the agents' values for the rooms, $\widehat{v}_i(r)=k=0$. 
Let us define the matchings of type $\mu_1$ to be the matchings where $(a_1,a_3)$ and $(a_2,a_4)$ are matched: $\mu_1 = \{(a_1,a_3, \cdot), (a_2,a_4, \cdot)\}$ where $\cdot$ indicates any disjoint assignment of the two rooms $r_1$ and $r_2$. Observe that, a matching of type $\mu_1$ has \USW{} zero. It can be easily verified that except for the matchings of type $\mu_1$, every other matching is welfare maximizing and yields a social welfare of $2$. 
However, each maximum \USW{} matching will satisfy only half of the agents. Thus, the other agents have an incentive to misreport their values. We show that an agent has incentive to misreport for each maximum welfare matching.

If the agents were matched as $\{(a_1,a_2,\cdot),(a_3,a_4,\cdot)\}$, $a_2$ can misreport her value for $a_3$ as $\beta$, where $\beta>\frac{2}{\alpha}-1$, such that matching $\{(a_1,a_4,\cdot),(a_2,a_3,\cdot)\}$ will generate social welfare $\beta+1$ such that that not only is it the maximum \USW{} roommate matching, but we also have that $\alpha(\beta +1)>2$. That is, any \USW{$\alpha$} mechanism must output $\{(a_1,a_4,\cdot),(a_2,a_3,\cdot)\}$. 
However in this case, $a_1$ could misreport his value for $a_2$ to be at least $\frac{2}{\alpha}-1$ analogously so as to be matched to $a_2$ himself. As a result, no \USW{$\alpha$} can be strategyproof in this instance. 

\paragraph{Leontief Utilities.} We now build the case for Leontief utilities. We set the true room values as $k=1$. Analogous to the proof for additive utilities, 
a matching of type $\mu_1$ has \USW{} of $0$. Further, the maximum \USW{} is $2$.
Similar to the previous arguments,  agents who get utility $0$ in the chosen matching can now misreport both their agent compatibility and room values as $\beta$ so the maximum \USW{} now becomes $\beta+1$ with $\beta>\frac{2}{\alpha}-1$. 
\end{proof}

 In contrast to \cref{thm:no-alpha-approx}, under binary valuations, approximate \USW{} mechanisms can be strategyproof. We now provide an upper bound to these approximations for binary additive utilities. 

\begin{restatable}{theorem}{binaddsp}\label{thm:addsp23}
    There exists no strategyproof mechanism that can guarantee \USW{$\alpha$} for any $\alpha>\sfrac{2}{3}$ , for binary valuations under additive utilities.
\end{restatable}

    

\begin{proof} 
First consider an instance with binary additive utilities as depicted in \cref{fig:binaddSP}. This, again, is an instance with $4$ agents and $2$ rooms. We first assume that $a_2$ has value only for $r_1$ and no value for any agent. That is, $v_{a_2}(i)=0=\widehat{v}_{a_2}(r_2)$ for all $i\neq 2$ and $\widehat{v}_{a_2}(r_1)=1$. 

Agents $a_1$ and $a_3$ have value only for $a_2$ and no value for any room. That is, $v_{a_1}(a_2)=v_{a_3}(a_2)=1$ and $v_{a_1}(i)=v_{a_3}(i')=\widehat{v}_{a_1}(r)=\widehat{v}_{a_3}(r)=0$ for any $i,i'\neq a_2$ and any room $r\in \R$. Agent $a_4$ has no value for any agent or room, that $v_{a_4}(i)=\widehat{v}(r)=0$ for all agents $i$ and rooms $r$. Observe that these valuations are not symmetric.

\begin{sloppypar}
There are two maximum \USW{} roommate matchings: $\mu_1 = \{(a_1,a_2,r_1)$, $(a_3,a_4,r_2)\}$ and $\mu_2=\{(a_1,a_4,r_2),(a_2,a_3,r_1)\}$. Both give \USW{} equal to $2$. Now if $a_1$ were to misreport her value for $r_1$ as being $1$, while all others were honest the only maximum roommate matching would be $\mu_1$ with \USW{}$=3$. In fact, $\mu_1$ would be the only \USW{$\alpha$} matching for $\alpha > \sfrac{2}{3}$. 

Analogously, if all other agents were honest, then $a_3$ could misreport his value for $r_1$ as $1$, and $\mu_2$ would be the only \USW{$\alpha$} matching for $\alpha > \sfrac{2}{3}$. Consequently, both $a_1$ and $a_3$ have incentives to misreport their value for $r_1$ as $1$ under any mechanism which is \USW{$\alpha$}. \end{sloppypar}

   \begin{figure}[t]
    \centering
    \tikzset{every picture/.style={line width=0.75pt}} 

\begin{tikzpicture}[x=0.85pt,y=0.75pt,yscale=-1,xscale=1]

\draw   (48,2.75) -- (53,11.75) -- (43,11.75) -- cycle ;
\draw   (3.25,57.71) .. controls (3.25,54.74) and (5.66,52.33) .. (8.63,52.33) .. controls (11.59,52.33) and (14,54.74) .. (14,57.71) .. controls (14,60.68) and (11.59,63.08) .. (8.63,63.08) .. controls (5.66,63.08) and (3.25,60.68) .. (3.25,57.71) -- cycle ;
\draw   (123.08,57.88) .. controls (123.08,54.91) and (125.49,52.5) .. (128.46,52.5) .. controls (131.43,52.5) and (133.83,54.91) .. (133.83,57.88) .. controls (133.83,60.84) and (131.43,63.25) .. (128.46,63.25) .. controls (125.49,63.25) and (123.08,60.84) .. (123.08,57.88) -- cycle ;
\draw   (83.25,57.71) .. controls (83.25,54.74) and (85.66,52.33) .. (88.63,52.33) .. controls (91.59,52.33) and (94,54.74) .. (94,57.71) .. controls (94,60.68) and (91.59,63.08) .. (88.63,63.08) .. controls (85.66,63.08) and (83.25,60.68) .. (83.25,57.71) -- cycle ;
\draw   (43.25,57.71) .. controls (43.25,54.74) and (45.66,52.33) .. (48.63,52.33) .. controls (51.59,52.33) and (54,54.74) .. (54,57.71) .. controls (54,60.68) and (51.59,63.08) .. (48.63,63.08) .. controls (45.66,63.08) and (43.25,60.68) .. (43.25,57.71) -- cycle ;
\draw   (128,3.25) -- (133,12.25) -- (123,12.25) -- cycle ;
\draw    (14,57.71) -- (43.25,57.71) ;
\draw    (48.63,52.33) -- (48.2,11.4) ;
\draw    (54,57.71) -- (83.25,57.71) ;
\draw [color={rgb, 255:red, 208; green, 10; blue, 27 }  ,draw opacity=1 ] [dash pattern={on 2pt off 1pt}]  (43,11.75) -- (8.63,52.33) ;
\draw [color={rgb, 255:red, 208; green, 10; blue, 27 }  ,draw opacity=1 ] [dash pattern={on 2pt off 1pt}]  (52.5,11.75) -- (88.63,52.33) ;

\draw (26,5) node [anchor=north west][inner sep=0.75pt]  [font=\scriptsize]  {$r_{1}$};
\draw (137,5) node [anchor=north west][inner sep=0.75pt]  [font=\scriptsize]  {$r_{2}$};
\draw (3,65) node [anchor=north west][inner sep=0.75pt]  [font=\scriptsize]  {$a_{1}$};
\draw (124,65) node [anchor=north west][inner sep=0.75pt]  [font=\scriptsize]  {$a_{4}$};
\draw (85,65) node [anchor=north west][inner sep=0.75pt]  [font=\scriptsize]  {$a_{3}$};
\draw (43,65) node [anchor=north west][inner sep=0.75pt]  [font=\scriptsize]  {$a_{2}$};

\end{tikzpicture}

    \caption{Agent preferences for instance described in \cref{thm:addsp23}. The optimal matching is shown in solid black edges and the misreports are shown in dashed red edges.}
    \label{fig:binaddSP}
\end{figure}
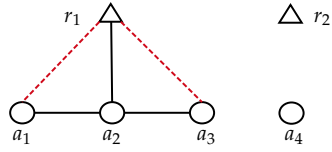 

\paragraph{Symmetric Valuations.} Now consider  a very similar instance, now with binary and symmetric valuations for the agents as depicted in \cref{fig:binaddSP}. That is, agent $a_2$ now has value for both $a_1$ and $a_3$, $v_{a_2}(a_1)=v_{a_2}(a_3)=1$. All other value remain the same as the earlier instance. Again, $\mu_1$ and $\mu_2$ (as defined earlier) are the only maximum \USW{} matchings, now with \USW{}$=3$. 

If $a_1$ were to misreport the value for $r_1$ as being $1$, then any \USW{$\alpha$} mechanism for $\alpha >\sfrac{3}{4}$ would always output $\mu_1$, if all agents other than $a_1$ were honest. Analogous to the earlier proof, no \USW{$\alpha$} mechanism can be strategyproof for $\alpha>\sfrac{3}{4}$.
\end{proof}

\begin{corollary}\label{cor:addspsymm}
There exists no strategyproof mechanism that can guarantee \USW{$\alpha$} for any $\alpha>\sfrac{3}{4}$, when agent valuations are symmetric,  in addition to being binary and additive.
\end{corollary}

A typical approach to finding strategyproof mechanisms is to use serial dictatorship where agents pick a roommate-room pair one by one from the remaining pool according to a fixed priority ordering. In \cref{app:sdessay} we discuss the standard serial dictatorship mechanisms and  various extensions and show how they either fail to give a guarantee on the \USW{} or fail to be strategyproof.

\begin{remark}
    The naive serial dictatorship mechanism is strategyproof but cannot guarantee \USW{$\alpha$} for any $\alpha > 0$ under both additive or Leontief utilities.
\end{remark}


\subsection{Strategyproof Mechanisms for Maximum Welfare}\label{subsec:spleon}


We now consider strategyproofness under binary Leontief utilities. In this setting, agents are constrained in how they can change the apparent social welfare of a matching by misreporting their values. Surprisingly, we show that unlike the additive setting, for binary Leontief utilities, there exist strategyproof mechanisms that guarantee maximum \USW{}. We shall design two  such mechanisms. First, in \cref{alg:doitall} we find all maximum welfare matchings and then select a matching from them in a manner that ensures strategyproofness. We then refine this approach to \cref{alg:doless} which is FPT parameterized by the number of agents. 
Both these mechanism rely on the structure provided by binary Leontief utilities. We begin with a simple observation that constrains the way agents can manipulate under binary Leontief utilities.

 \begin{observation}\label{obs:binleonlie}
    Let  $i\in \N$ denote an agent, and $
    \mu_1$ and $\mu_0$ denote two matchings where utility of agent $i$ is $1$ and $0$, respectively, under binary Leontief utilities. Then, when all other agents $N \setminus \{i\}$ are honest
    \begin{enumerate}
        \item agent $i$ cannot increase the \USW{} of $\mu_1$ by misreporting.
        \item agent $i$ cannot decrease the \USW{} of $\mu_0$ by misreporting.
    \end{enumerate}
    \end{observation}

While this does not guarantee that there are never any incentives to lie, it proves to be useful in the development of strategyproof mechanisms that maximize social welfare. 

From \cref{thm:SW-hardL} we know that any maximum \USW{} mechanism cannot run in polynomial time even for restricted binary valuations. Consequently, we first consider a brute force approach where we look at all possible matchings and choose those that maximize \USW{}. Tie-breaking rules are crucial to achieve strategyproofness.
We show in Appendix \ref{app:manipulationexamples} that there exist methods of tie-breaking that give some agent an incentive to lie about their preferences. Essentially, agents get incentivized to manipulate when the tie-breaking rule to choose one among the maximum welfare matchings is a function of agent preferences. 
A careful consideration of tie-breaking through a priority ordering over agents recovers strategyproofness.

\begin{algorithm}[t]
    {\small
  \KwIn{ Roommate Matching instance $\langle \N,\R, v,\widehat{v} \rangle$ with binary valuations}
  \KwOut{A matching $\mu$ with maximum \USW{}}
  Initialize $S_0$ to be the set of all roommate matchings with maximum \USW{} \;
  Pick an ordering on agents $\sigma$\;
  \For{$t=1$ to $2n$}{
        Select $j$ to be the $i^{\text{th}}$ agent on $\sigma$\;
        $S_i\gets \argmax_{\mu \in S_{i-1}}u_j(\mu)$ \Comment{\small Reduce based on $j$}\;
    }
  Select $\mu$ arbitrarily from $S_{2n}$\; 
  \caption{Welfare Set Reduction Mechanism} \label{alg:doitall}}
\end{algorithm}

\paragraph{Overview of Mechanism.} Algorithm \ref{alg:doitall} proceeds by first computing the set of all roommate matchings with maximum \USW{}: $S_0$. Then we fix an {\em arbitrary} ordering of agents $\sigma$. We use $i = \sigma(j)$ to denote the position of the agent $j$. We sequentially define the set $S_i=\argmax_{\mu \in S_{i-1}} u_{j}(\mu)$. 
Finally, the mechanism returns a matching from the set $S_{2n}$. We first give an example of its execution.

\begin{example}
    We use the instance in \cref{thm:no-alpha-approx} to show an execution of \cref{alg:doitall}. Recall that here all agents have value $1$ for each room, but are selective over who they are matched to. The agent compatibility values are represented graphically in \cref{fig:noapproxSW}. Agents $a_1$ and $a_3$ would like to be matched with agents $a_2$ and $a_4$ respectively. Whereas, agents $a_2$ and $a_4$ would like to be matched to $a_3$ and $a_1$, respectively. 
    
    Let the ordering $\sigma = a_1\succ a_2\succ a_3\succ a_4$. Recall that the maximum \USW{} matchings were either matching agents as $\mu=\{(a_1,a_2,\cdot),(a_3,a_4,\cdot)\}$ or as $\mu'=\{(a_1,a_4,\cdot),(a_2,a_3,\cdot)\}$. Therefore, in $S_0$ agents are matched as in $\mu$ or as in $\mu'$ with all possible options of room assignments. Agent $a_1$ gets utility $1$ only when matched to $a_2$. Thus, $a_1$ deletes the matchings where agents are matched as in $\mu'$ from $S_0$ to construct $S_1$. 
    
    Now $a_2$ gets utility $0$ when matched to $a_1$ and consequently in both the matchings in $S_1$. Thus, $S_2= S_1$. 
    Finally $S_4$ will contain the two matchings from $S_1$.
    Thus, the algorithm arbitrarily picks between $\{(a_1,a_2,r_1),(a_3,a_4,r_2)\}$ and $\{(a_1,a_2,r_2),(a_3,a_4,r_1)\}$ the two matchings in $S_4$. This is illustrated in \cref{fig:doitallexample}.
    Note that $a_2$ and $a_4$ get utility $0$ in either case. There is no way $a_2$ (or $a_4$) can misreport and change the set $S_2$  as the matchings where $a_2$ (resp. $a_4$) get utility one has been already deleted in $S_1$.

    Also, observe that if the ordering had first selected $a_2$ or $a_4$ the matching selected would arbitrarily be picked from the two matchings of type $\mu'$. 

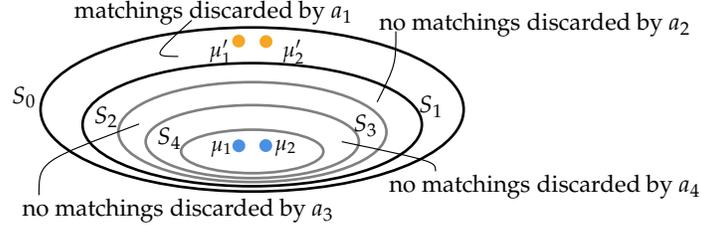
\begin{figure}[t]
    \centering

\tikzset{every picture/.style={line width=1pt}}     

\begin{tikzpicture}[x=0.75pt,y=0.75pt,yscale=-1,xscale=1]
\draw   (20.8,44.01) .. controls (20.8,21.14) and (68.57,2.6) .. (127.5,2.6) .. controls (186.43,2.6) and (234.2,21.14) .. (234.2,44.01) .. controls (234.2,66.87) and (186.43,85.41) .. (127.5,85.41) .. controls (68.57,85.41) and (20.8,66.87) .. (20.8,44.01) -- cycle ;
\draw  [color={rgb, 255:red, 245; green, 166; blue, 35 }  ,draw opacity=1 ][fill={rgb, 255:red, 245; green, 166; blue, 35 }  ,fill opacity=1 ] (118.13,9.47) .. controls (118.13,8.05) and (119.29,6.89) .. (120.71,6.89) .. controls (122.13,6.89) and (123.29,8.05) .. (123.29,9.47) .. controls (123.29,10.89) and (122.13,12.05) .. (120.71,12.05) .. controls (119.29,12.05) and (118.13,10.89) .. (118.13,9.47) -- cycle ;
\draw   (41.99,51.67) .. controls (41.99,34.51) and (80.3,20.6) .. (127.55,20.6) .. controls (174.81,20.6) and (213.12,34.51) .. (213.12,51.67) .. controls (213.12,68.83) and (174.81,82.74) .. (127.55,82.74) .. controls (80.3,82.74) and (41.99,68.83) .. (41.99,51.67) -- cycle ;
\draw  [color={rgb, 255:red, 128; green, 128; blue, 128 }  ,draw opacity=1 ] (60.03,55.41) .. controls (60.03,41.49) and (90.3,30.2) .. (127.64,30.2) .. controls (164.98,30.2) and (195.25,41.49) .. (195.25,55.41) .. controls (195.25,69.33) and (164.98,80.62) .. (127.64,80.62) .. controls (90.3,80.62) and (60.03,69.33) .. (60.03,55.41) -- cycle ;
\draw  [color={rgb, 255:red, 128; green, 128; blue, 128 }  ,draw opacity=1 ] (73.82,60.23) .. controls (73.82,50.05) and (97.89,41.8) .. (127.57,41.8) .. controls (157.26,41.8) and (181.33,50.05) .. (181.33,60.23) .. controls (181.33,70.4) and (157.26,78.65) .. (127.57,78.65) .. controls (97.89,78.65) and (73.82,70.4) .. (73.82,60.23) -- cycle ;
\draw  [color={rgb, 255:red, 128; green, 128; blue, 128 }  ,draw opacity=1 ] (91.68,65.18) .. controls (91.68,59.11) and (107.75,54.2) .. (127.57,54.2) .. controls (147.38,54.2) and (163.45,59.11) .. (163.45,65.18) .. controls (163.45,71.24) and (147.38,76.15) .. (127.57,76.15) .. controls (107.75,76.15) and (91.68,71.24) .. (91.68,65.18) -- cycle ;
\draw  [color={rgb, 255:red, 245; green, 166; blue, 35 }  ,draw opacity=1 ][fill={rgb, 255:red, 245; green, 166; blue, 35 }  ,fill opacity=1 ] (131.94,9.71) .. controls (131.94,8.28) and (133.09,7.13) .. (134.52,7.13) .. controls (135.94,7.13) and (137.09,8.28) .. (137.09,9.71) .. controls (137.09,11.13) and (135.94,12.28) .. (134.52,12.28) .. controls (133.09,12.28) and (131.94,11.13) .. (131.94,9.71) -- cycle ;
\draw  [color={rgb, 255:red, 74; green, 144; blue, 226 }  ,draw opacity=1 ][fill={rgb, 255:red, 74; green, 144; blue, 226 }  ,fill opacity=1 ] (118.23,62.33) .. controls (118.23,60.91) and (119.39,59.75) .. (120.81,59.75) .. controls (122.23,59.75) and (123.38,60.91) .. (123.38,62.33) .. controls (123.38,63.75) and (122.23,64.91) .. (120.81,64.91) .. controls (119.39,64.91) and (118.23,63.75) .. (118.23,62.33) -- cycle ;
\draw  [color={rgb, 255:red, 74; green, 144; blue, 226 }  ,draw opacity=1 ][fill={rgb, 255:red, 74; green, 144; blue, 226 }  ,fill opacity=1 ] (131.82,62.22) .. controls (131.82,60.79) and (132.97,59.64) .. (134.39,59.64) .. controls (135.82,59.64) and (136.97,60.79) .. (136.97,62.22) .. controls (136.97,63.64) and (135.82,64.79) .. (134.39,64.79) .. controls (132.97,64.79) and (131.82,63.64) .. (131.82,62.22) -- cycle ;

\draw (104.82,57.96) node [anchor=north west][inner sep=0.75pt]  [font=\scriptsize]  {$\mu _{1}$};
\draw (141.19,7.21) node [anchor=north west][inner sep=0.75pt]  [font=\scriptsize]  {$\mu '_{2}$};
\draw (104.25,7.12) node [anchor=north west][inner sep=0.75pt]  [font=\scriptsize]  {$\mu '_{1}$};
\draw (137.14,57.6) node [anchor=north west][inner sep=0.75pt]  [font=\scriptsize]  {$\mu _{2}$};
\draw (4.72,30.55) node [anchor=north west][inner sep=0.75pt]  [font=\footnotesize]  {$S_{0}$};
\draw [draw opacity=1, line width=0.15mm ]   (90.25,17.21) .. controls (77,17.21) .. (97,-1) ;
\draw (36.17,-12.6) node [anchor=north west][inner sep=0.75pt]  [font=\footnotesize]  {matchings discarded by $a_1$};
\draw (210,35.23) node [anchor=north west][inner sep=0.75pt]  [font=\footnotesize]  {$S_{1}$};
\draw (46.17,40.71) node [anchor=north west][inner sep=0.75pt]  [font=\footnotesize]  {$S_{2}$};
\draw [draw opacity=1,line width=0.15mm]   (190.55,40.23) .. controls (200,21.21) .. (200,4.6) ;
 \draw (190,-6.6) node [anchor=north west][inner sep=0.75pt]  [font=\footnotesize]  {no matchings discarded by $a_2$};
\draw (177,45) node [anchor=north west][inner sep=0.75pt]  [font=\footnotesize]  {$S_{3}$};
\draw [draw opacity=1, line width=0.15mm ]  (70.51,51.32) .. controls (20,70.21) .. (20,87.21) ;
 \draw (10,87.21) node [anchor=north west][inner sep=0.75pt]  [font=\footnotesize]  {no matchings discarded by $a_3$};
\draw (78,51.32) node [anchor=north west][inner sep=0.75pt]  [font=\footnotesize]  {$S_{4}$};
\draw [draw opacity=1, line width=0.15mm ]   (172.71,60.25) .. controls (220,70.21) .. (220,75.21) ;
 \draw (195,75.21) node [anchor=north west][inner sep=0.75pt]  [font=\footnotesize]  {no matchings discarded by $a_4$};

\end{tikzpicture}

    {\small
    \caption{\Cref{alg:doitall} on the instance depicted in \cref{fig:noapproxSW}. }
    \label{fig:doitallexample}}
\end{figure}
  
\end{example}


\begin{restatable}{theorem}{spmechone}\label{prop:spmech1}
    The Welfare Set Reduction mechanism (\cref{alg:doitall})  returns a matching with maximum \USW{} and is strategyproof under binary Leontief utilities.
\end{restatable}

\begin{proof}

Firstly, it is straightforward to see that \cref{alg:doitall} always returns a maximum \USW{} roommate matching. It first selects all such matchings and then chooses among them.  We shall now show that this mechanism is strategyproof. 

Fix an arbitrary instance with binary Leontief utilities $I=\langle \N,\R, v, \widehat{v}\rangle$. Observe that $S_0\supseteq S_1 \supseteq \dots \supseteq S_{2n}$. Let $j\in \N$ be the $i^{th}$ agent in the order $\sigma$. We shall show that irrespective of whether the sets $S_{i-1}$ and $S_0$ contain a matching where $j$ gets utility $1$, $j$ has no incentive to lie.

\paragraph{Case 1.} Let $S_{i-1}$ contain at least one matching where $j$ gets utility 1. Then, we have that $S_i$ will only contain matchings that give $i$ utility $1$. As $S_{2n}\subseteq S_i $ then the mechanism will always output a matching which gives $j$ utility $1$. Thus, in this case, $j$ has no incentive to lie.  

\paragraph{Case 2.} Let $S_{0}$ contain no matchings where $j$ gets utility $1$. By \cref{obs:binleonlie}, $j$ cannot introduce any matchings to $S_0$ where $j$ truly gets utility $1$. Similarly, $j$ cannot remove matchings from $S_0$, as where $j$ truly gets utility $0$ for all matchings in $S_0$.  By misreporting, $j$ can either 1) reduce the \USW{} of matchings in which $j$ truly gets utility $1$ or 2) increase the \USW{} of matchings where $j$ truly gets $0$. Consequently, whether $j$ is honest or misreports, the matching selected will be such that $j$ has utility $0$. Hence, $j$ has no incentive to misreport preferences. 

\paragraph{Case 3.} Let $S_0$ and $S_{i-1}$ be such that  $S_{0}$ contains at least one matchings where $j$ gets utility $1$, but $S_{i-1}$ contains no matchings where $j$ gets utility $1$. In this case, a set of agents with higher precedence to $j$ were prioritized over $j$ in selecting the matchings in $S_{i-1}$. That is, for every matching in $S_0$ where $j$ gets utility $1$, an agent with higher precedence to $j$ gets utility $0$. Thus, the corresponding set does not contain the matching.  Again, by \cref{obs:binleonlie}, 1) $j$ cannot decrease the \USW{} of the matchings in $S_{i-1}$, 2) $j$  cannot introduce a new matching that is in $S_0 \setminus S_{i-1}$ to the set $S_{i-1}$ such that  $j$ truly gets utility $1$ in that matching.    
    
By misreporting, $j$ can either a) reduce the \USW{} of matchings in which $j$ truly gets utility $1$ or b) increase the \USW{} of matchings where $j$ truly gets $0$. To show that $j$ does not have an incentive to misreport, we shall now show that neither of these prevents the mechanism from discarding matchings where $j$ truly gets $1$. 

\paragraph{Case 3a.} Suppose that agent $j$ misreports such that the (perceived) \USW{} of a set $X$ of matchings is increased under the misreported valuation. From \cref{obs:binleonlie}, we have that $j$ truly gets $0$ in each matching in $X$. This can affect the outcome of the mechanism if: 

\begin{enumerate}
    \item A matching in $X$ belongs to $S_0$ under the true valuations, or
    \item the true social welfare of a matching in $X$ is one less than the maximum \USW{}, i.e., the social welfare of the matchings in $S_0$ under true valuation.
\end{enumerate}

Increasing the social welfare of matchings already in $S_0$ will only ensure that $j$ receives utility $0$. On the other hand, effectively adding new matchings to $S_0$ can only introduce more matchings to $S_{i-1}$ where $j$ gets utility $0$. Hence, $j$ has no incentive to increase the apparent welfare of matchings that give $j$ utility $0$.

\paragraph{Case 3b.} The remaining case is to show that  decreasing the \USW{} of a set $X$ of matchings where $j$ truly gets utility $1$ is also not beneficial to $j$. Recall that $S_0$ contains at least one roommate matching that gives $j$ utility $1$, while $S_{i-1}$ only contains matchings that give $j$ utility $0$. 

Let $j_1, \cdots, j_k$ be the agents (in order of precedence) who eliminate the matchings that gives $j$ utility $1$, where $j_1$ is the highest precedence agent that does this, $j_2$ is the second, and so on. It is clear from the algorithm that each agent $j_\ell$, for $\ell \in [k]$, must all receive utility $1$ from the matchings in $S_{i-1}$. Further, $k\geq 1$ since $j$ receives utility $1$ in at least one matching in $S_0$ that is discarded before the algorithm reaches $j$. 

If $k=1$, then $j_1$ eliminates all matchings where $j$ gets utility $1$, and thus, there is no maximum \USW{} matching where both $j$ and $j_1$ get utility $1$. Consequently, by misreporting and decreasing the \USW{} of any matching where $j$ gets utility $1$, it will not effect the contents of $S_{\sigma (j_1)}$. Hence, if $k=1$, misreporting preferences cannot benefit $j$.

If $k>1$, clearly, $j_1$ has multiple matchings which gives $j_1$ utility $1$, at least one of which gives utility $0$ to $j$ and utility $1$ to $j_1,\cdots, j_k$. Let this matching be denoted by $\mu$. Again, decreasing the apparent \USW{} of any  matching which gives $j$ utility  $1$ will not affect $\mu$ being contained in $S_{i-1}$ and neither does it prevent the elimination of any other matching where $j$ truly gets utility $1$. 

Consequently, as long as there is a set of agents $j_1,\cdots, j_k$, each with higher precedence to $j$, such that any maximum \USW{} roommate matching which gives $j$ utility $1$, gives at least one of $j_1,\cdots, j_k$ utility $0$, the set $S_{i-1}$ will continue to only have matchings where $j$ gets utility $0$. 

Hence, $j$ cannot misreport and ensure that they receive utility $1$ in either Case 3a or 3b.\\

\noindent As a result, no agent has an incentive to lie under this mechanism. 
\end{proof}



\begin{remark}
 Observe that to show \Cref{alg:doitall} is strategyproof, we need to use \Cref{obs:binleonlie} which holds only for binary valuations. Moreover, we showed in \Cref{thm:no-alpha-approx,thm:addsp23} that no  algorithm that guarantees maximum welfare is strategyproof when the preferences are not binary. Thus, we can not hope to find a strategyproof algorithm for general valuations.
\end{remark}
Both \cref{obs:binleonlie} and the proof of \cref{prop:spmech1} rely on the fact that agents utilities can only be $0$ or $1$ and are not contingent on the utilities being specifically Leontief. 
Consequently, the following holds. 

\begin{corollary}
    If for each agent $i \in \N$, the utility $u_i(\mu) \in \{0,1\}$, then there exists a welfare maximizing matching $\mu$ that is strategyproof under any utility function.
\end{corollary}

Although the Welfare Set Reduction mechanism (\cref{alg:doitall}) shows the existence of a strategyproof mechanism, it is far from being efficient, even for an exponential time algorithm. In particular, it must parse the set of all max welfare matchings which can have size up to $n^{2n}$. We now build on the ideas in this mechanism to produce a strategyproof mechanism that is parameterized by the number of agents.

\paragraph{Overview of Precedence-Based Search Mechanism.}
In \Cref{alg:doless}, we fix an ordering $\sigma$ of the agents. We iterate over the (possible) value of the maximum \USW{} starting from $2n$ to $1$. For each value $k$ of \USW{}, consider the subsets of agents $\N$ of size $k$ in the precedence ordering of them implied by $\sigma$.\footnote{For two subsets $T$ and $T'$ of the same size, we say that $T$ has higher precedence than $T'$, if there is an agent $i\in T \setminus T'$ which has higher precedence than all agents in $T'\setminus T$.} 
For each subset, we check if there exists a matching that gives utility one to each agent in this subset and gives utility zero to each agent not in the subset. 

In order to do this, we construct an instance of (unweighted) $3$-set packing $(\mathcal{U}, \tau)$ for a subset $T$ as follows. We define $\mathcal{U}$ to be $\N\cup \R$, for each agent $i \in T$. We choose $\tau$ to be the set of triples $(i,j,r)$ s.t. either agent gets utility $1$ from the triple if and only if they are in the subset $T$. 
Based on this universe of elements $\mathcal{U}$ and collection of sets $\tau$, we find a $3$-set packing of size $n$ if and only if a roommate matching exists which gives the agents in $T$ utility $1$ and those not in $T$ utility $0$. If no set packing exists for any choice of $k$ and  $T$, we return an arbitrary matching, inferring that all roommate matching have \USW{} of $0$.

\begin{algorithm}[t]{\small 
  \KwIn{ Roommate Matching instance $\langle \N,\R, v,\widehat{v} \rangle$ with binary valuations}
  \KwOut{A matching $\mu$ with maximum \USW{}}
  Pick an arbitrary ordering on agents $\sigma$\;
  $k \gets 2n$\;
  \While{$k\geq 1$}{
    Initialize $T_1\cdots T_{\binom{2n}{k}}$ s.t. $T_t$ is the $t^{\text{th}}$ highest precedence subset of $\N$ of size $k$ under $\sigma$\;
    
    \For{$t=1$ to $\binom{2n}{k}$}{
        $\tau \gets \{(i,j,r)| i \neq j \in \N, r \in \R$ s.t. $u_i(\{j,r\})=1 \Leftrightarrow \, i\in T_t$ and
        $u_j(\{i,r\})=1 \Leftrightarrow \, j\in T_t\}$ \;
        $\mu \gets$ \textsc{$3$-Set Packing}$(\N \cup \R, \tau)$\;
        \If{$|\mu|=n$}{
            \textbf{Return} $\mu$\;
        }
    
  }
    $k\gets k-1$\;
  }
  \textbf{Return} an arbitrary matching\;
   \caption{Precedence-Based Search Mechanism}\label{alg:doless}}
\end{algorithm}
\begin{theorem}\label{thm:spmech2}
There exists a strategyproof algorithm to find a maximum \USW{} roommate matching that runs in time $O^*(c^n)$ where $c$ is a constant. 
\end{theorem}

\begin{proof}
We show that given a roommate matching instance $I=\langle \N,\R, v,\widehat{v} \rangle$ with binary Leontief utilities, \cref{alg:doless} always returns a roommate matching with maximum \USW{}, is strategyproof, and runs in time $O^*(c^n)$.
    Suppose that \cref{alg:doless} returns a matching $\mu$ for $k =w$, i.e., $\USW{}(\mu) = w$.
    
\textbf{Welfare Guarantee:} 
 We first argue that \cref{alg:doless} produces a matching with maximum \USW{}. Recall that for a subset of agents $T$, $\tau$ is the set of triples for which agents in $T$ get utility $1$ and those not in $T$ get utility $0$.
Clearly, whenever a roommate matching $\mu$ of \USW{} $w$ exists, exactly $w$ agents must receive a utility one, let set this be $T$. All other agents get a utility zero from it. This would imply that it is a set packing of size $n$ for $\tau$ defined for $T$. In particular, the triples of $\mu$ would be the set packing for $\tau$. Analogously, whenever a set packing of size $n$ exists for a set $T$ of size $w$, those triples form a roommate matching with \USW{} equal to $w$. 

Consequently, the algorithm looks for roommate matchings of \USW{} $2n$ to $1$. Since we iterate over all possible values of welfare and all possible sets of agents to achieve the welfare, clearly, we find the maximum \USW{} matching.

\textbf{Strategyproofness:}
We show that no agent has an incentive to misreport their values. To this end, we prove that the outcome of this mechanism would be the same as that of 
the mechanism in \cref{alg:doitall} for the same ordering $\sigma$ of agents. For a given instance $I$, let $\mu$ be the matching returned, having \USW{} $w$. Recall the sets $S_0,\cdots, S_{2n}$ as defined in \cref{alg:doitall}. As $\mu$ has maximum \USW{}, $\mu \in S_0$. We shall show that $\mu\in S_{2n}$ for ordering $\sigma$. We shall later show that this is sufficient to guarantee strategyproofness. 

Let $T$ be the set of $w$ agents who get utility $1$ in $\mu$. Let the positions of the  agents in $T$ be denoted by $j_1, j_2, \dots, j_w$ where $j_1< j_2< \dots < j_w$ in the ordering $\sigma$.
To prove that $\mu \in S_{2n}$,  it suffices to show that $\mu \in S_{j_w}$. As the maximum possible $\USW{}$ is $w$, every agent after $j_w$ in the ordering $\sigma$ receives utility $0$ from every matching in $S_{j_w}$. Consequently, $S_{2n}$ is the same as $S_{j_w}$ and \cref{alg:doitall} would also return $\mu$. 

In order to show that $\mu \in S_{j_w}$, we need to show that no agent who comes before position $j_w$ would be able to remove $\mu$. In particular, we need to show that if an agent $i \notin T$   (i) receives utility $1$ under some maximum \USW{} matching $\mu_i$ and (ii) under $\sigma$, precedes the agent at $j_w$, then $i$ has utility $0$ from all matchings in $S_{\sigma(i)-1}$. Thus, $S_{\sigma(i)}=S_{\sigma(i)-1}$ and $\mu \in S_{\sigma(i)}$. 
We show this by induction on $\sigma(i)$.

\underline{Base Case:}  Firstly, we consider the case where $i$ precedes the agent at $j_1$, i.e., $\sigma(i) < j_1$. Let $T'$ be the set of $w$ agents who receive utility $1$ in a matching $\mu_i$ where $i$ receives utility $1$. Then, $T'$ precedes $T$ as $p$ precedes all  agents in $T$. Then \cref{alg:doless} would return $T'$ instead of $T$. Thus, no such $i$ can exist.
Hence, every agent $i'$ that occurs before $j_1$ must have utility $0$ for all matchings in $S_{\sigma(i')-1}$. Consequently, for each agent  at a position $i'$, which is before $j_1$, $S_{\sigma(i')}=S_{\sigma (i')-1}$ and $\mu \in S_{\sigma(i')}$.  Thus, $\mu\in S_{j_1}$.

\underline{Induction Case:} 
Next, fix a value $\ell$ s.t. $1<\ell\leq w$ and assume that $\mu \in S{j_{\ell -1}}$. 
We now show that 
the matching $\mu \in S_{j_\ell}$. In particular, $\mu \in S_{\sigma (i')}$ for all agents $i$ that lie between $j_{\ell -1}$ and $j_{\ell}$. As in the base case, we shall prove this by contradiction.  

Assume that there is an agent $i$ that lies between $j_{\ell -1}$ and $j_{\ell}$ and $\mu_i$ denotes a {\em maximum} \USW{} matching for which $i$ has utility $1$. That is, $\USW{}(\mu)=w$. Note that if $\mu_i$ is not contained in $S_{\sigma(i)-1}$, we have from the proof of \cref{prop:spmech1} that $i$ cannot misreport and include it. Consequently, it only remains to argue for the case when it is.

Let $\mu_i$ be contained in $S_{\sigma(i)-1}$. By the induction hypothesis, $\mu\in S_{j_{\ell-1}}$  and thus the agents  $j_{\ell-1},j_{\ell-2}, \cdots, j_1$ have utility $1$ in $\mu$. Consequently, it must be that agents $j_{\ell-1},j_{\ell-2}, \cdots, j_1$ have utility $1$ for $\mu_i$ since $j_{\ell-1} < \sigma(i)$ and we assume that $\mu_i\in S_{\sigma (i)}$. 
%
Let $T'$ be the set of agents that have utility $1$ in $\mu_i$. Clearly, $T'$ contains $i$ and $j_{\ell-1},j_{\ell-2}, \cdots, j_1$. Observe that as $i$ has higher precedence than $j_{\ell}$, and $T'$ has $w$ agents including all the agents with higher precedence to $j_\ell$ in $T$. As a result,  
the set $T'$ must have higher precedence than $T$ under $\sigma$. Analogous to the base case, this is not possible, as \cref{alg:doless} checks for $T'$ before $T$. 
Thus, there cannot exist such a matching $\mu_i$, and it must be that $\mu \in S_{j_{\ell}}$. 

Therefore, by induction, $\mu\in S_{j_w}$ and consequently, $\mu \in S_{2n}$.
Now observe that all matchings in $S_{2n}$ would give the agents in $T$ a utility of $1$ and all other agents a utility $0$. Hence, every agent is indifferent across all matchings in $S_{2n}$ and they don't have an incentive to misreport.

Finally, we calculate the time complexity of \cref{alg:doless}. 


\textbf{Running time:} For each value of the welfare $w$, the algorithm runs in time $\binom{2n}{w}\cdot O^*(3.34^n)$~\citep{bjorklund2017narrow} if we use a randomized algorithm for (unweighted) {\sc $3$-set packing}, or a deterministic algorithm in time $\binom{2n}{w}\cdot O^*(6.75^{n+o(n)})$~\citep{feng2014matching}. We iterate over all subsets of $[2n]$. Thus the runtime is $O^*(13.36^n)$ (randomized) or $O^*(27^n)$ (deterministic). Hence we prove the theorem.
\end{proof}



\subsection{Polynomial Time Strategyproof Mechanisms}\label{subsec:approxswsp}
We now demonstrate different polynomial time mechanisms that are strategyproof and give non-trivial guarantees on \USW{}. We first attempt this via a serial dictatorship approach. While serial dictatorship is a well used approach to achieve strategyproofness, in its typical form it proves to be ineffective in providing any meaningful guarantee on \USW{}. We discuss this at length in \cref{app:sdessay}, along with some unsuccessful modifications to the standard approach. However, we can adapt the L/T maximal matching approach we discussed in a serial dictatorship fashion to get a strategyproof mechanism, described in \cref{alg:LT-alg}, that is \USW{$\sfrac{1}{6}$} for binary Leontief utilities. We defer the proof of this to \cref{app:sdleon}. 
\begin{restatable}{theorem}{sdleon}\label{lem:sdleon}
    There exists a strategyproof mechanism (\cref{alg:LT-alg}) that is \USW{$\sfrac{1}{6}$} under binary Leontief utilities.
   
\end{restatable}

\begin{algorithm}[t]\small
  \KwIn{A roommate matching instance $\langle \N,\R, v,\widehat{v} \rangle$ }
  \KwOut{A matching $\mu$}
  Initialize matching $\mu \gets \emptyset$ \;
  Let $\sigma = (\sigma_1, \ldots, \sigma_{2n})$ be a priority ordering over agents \;
  \For{$i = 1$ to $2n$}{
    \If{agent $\sigma_i$ is not yet matched}{
        \If{agent $\sigma_i$ values any $(j, r)$ pair among remaining}{
            $\mu \gets (\sigma_i, j, r)$ \;
        }
    
    }
  }
  Match the remaining agents to any room arbitrarily\;
   \caption{L/T Serial Dictatorship Mechanism}\label{alg:LT-alg}
\end{algorithm}
\noindent In fact, an analogous mechanism also gives a similar guarantee for additive utilities. We defer the proof to \cref{app:sdessay}. 

\begin{restatable}{theorem}{sdstar}\label{thm:sdadd}
    The Welfare Prioritized Serial Dictatorship mechanism (Algorithm \ref{alg:SDstar}) is strategyproof under binary additive utilities. Further, it is $\USW{\sfrac{1}{7}}$. 
\end{restatable}


\paragraph{Triangle-then-L Mechanism.} 
We find that the Triangle-then-L mechanism provides us a strategyproof mechanism with a better approximation guarantee for binary Leontief utilities. Recall that this mechanism is \USW{$\sfrac{1}{3}$}.

\begin{restatable}{theorem}{greedymechanism}\label{thm:greedy}
The Triangle-then-L mechanism (\cref{alg:greedy}) is strategyproof under binary Leontief utilities. 
\end{restatable}

\begin{proof}

\begin{sloppypar}
Fix an arbitrary roommate matching instance $I=\langle \N,\R, v,\widehat{v} \rangle$. Let $\mu$ be the matching returned by Algorithm \ref{alg:greedy} on $I$. Let $\mu^*$ be a social welfare maximizing roommate matching on $I$. We know that this algorithm always returns a \USW{$\sfrac{1}{3}$} roommate matching. We now prove that the mechanism is strategyproof for binary Leontief utilities.
\end{sloppypar}

Assume that the given instance $I=\langle \N,\R, v,\widehat{v} \rangle$ is one where all agents have binary Leontief utilities.
Analogous to \cref{obs:binleonlie}, agents cannot lie and increase the \USW{} of triples in which they truly get utility $1$. Similarly, an agent cannot decrease the \USW{} of a triple in which they truly get utility $0$. 
    Further, the mechanisms uses a deterministic tie-breaking rule across the triples. As a result, the agents cannot misreport their preferences and ensure that a triple in which they truly receive utility $1$ is selected by the mechanism.  
\end{proof}


\section{Concluding Remarks}

This paper initiated the study of strategyproofness and Leontief utilities in the setting of roommate matchings. We showed that maximizing social welfare is intractable and incompatible with strategyproofness when valuations are unrestricted. Surprisingly, there exists a strategyproof mechanism that maximizes welfare given binary Leontief utilities. Subsequently, we devise such an FPT parameterized mechanism.
We complemented these results by designing strategyproof mechanisms that achieve constant factor approximations on the social welfare for both additive and Leontief utilities.


Our work paves the way for further work in characterizing strategyproof roommate-room matching mechanisms. Future directions include finding better approximation algorithms for Leontief utilities and/or polynomial time strategyproof mechanisms with better \USW{} guarantees for binary valuations. Another interesting direction is characterizing the set of strategyproof mechanisms or weakening the strateyproofness requirement (e.g. restricting manipulation) to improve the welfare.

\section*{Acknowledgments}
This research was supported in part by NSF Awards IIS-2144413 and IIS-2107173.
We thank the anonymous reviewers for their helpful comments.



{\small 
\bibliographystyle{named} 
\bibliography{references.bib}}

\begin{thebibliography}{}

\bibitem[\protect\citeauthoryear{Abdulkadiro{\u{g}}lu and
  S{\"o}nmez}{2003}]{abdulkadirouglu2003school}
Atila Abdulkadiro{\u{g}}lu and Tayfun S{\"o}nmez.
\newblock School choice: A mechanism design approach.
\newblock {\em American economic review}, 93(3):729--747, 2003.

\bibitem[\protect\citeauthoryear{Abraham \bgroup \em et al.\egroup
  }{2004}]{abraham2004pareto}
David~J Abraham, Katar{\'\i}na Cechl{\'a}rov{\'a}, David~F Manlove, and Kurt
  Mehlhorn.
\newblock Pareto optimality in house allocation problems.
\newblock In {\em International symposium on algorithms and computation}, pages
  3--15. Springer, 2004.

\bibitem[\protect\citeauthoryear{Ashlagi \bgroup \em et al.\egroup
  }{2015}]{ashlagi2015mix}
Itai Ashlagi, Felix Fischer, Ian~A Kash, and Ariel~D Procaccia.
\newblock Mix and match: A strategyproof mechanism for multi-hospital kidney
  exchange.
\newblock {\em Games and Economic Behavior}, 91:284--296, 2015.

\bibitem[\protect\citeauthoryear{Aziz \bgroup \em et al.\egroup
  }{2011}]{aziz2011optimal}
Haris Aziz, Felix Brandt, and Hans~Georg Seedig.
\newblock Optimal partitions in additively separable hedonic games.
\newblock In {\em Proceedings of the Twenty-Second International Joint
  Conference on Artificial Intelligence}, pages 43--48, 2011.

\bibitem[\protect\citeauthoryear{Aziz \bgroup \em et al.\egroup
  }{2012}]{aziz2012individual}
Haris Aziz, Paul Harrenstein, and Evangelia Pyrga.
\newblock Individual-based stability in hedonic games depending on the best or
  worst players.
\newblock In {\em Proceedings of the 11th International Conference on
  Autonomous Agents and Multiagent Systems-Volume 3}, pages 1311--1312, 2012.

\bibitem[\protect\citeauthoryear{Aziz \bgroup \em et al.\egroup
  }{2013}]{aziz2013pareto}
Haris Aziz, Felix Brandt, and Paul Harrenstein.
\newblock Pareto optimality in coalition formation.
\newblock {\em Games and Economic Behavior}, 82:562--581, 2013.

\bibitem[\protect\citeauthoryear{Aziz}{2013}]{aziz2013stable}
Haris Aziz.
\newblock Stable marriage and roommate problems with individual-based
  stability.
\newblock In {\em Proceedings of the 2013 international conference on
  Autonomous agents and multi-agent systems}, pages 287--294, 2013.

\bibitem[\protect\citeauthoryear{Aziz}{2021}]{aziz2021individually}
Haris Aziz.
\newblock Individually rational land and neighbor allocation: Impossibility
  results.
\newblock {\em arXiv preprint arXiv:2106.03000}, 2021.

\bibitem[\protect\citeauthoryear{B{\'e}rczi \bgroup \em et al.\egroup
  }{2022}]{berczi2022manipulating}
Krist{\'o}f B{\'e}rczi, Gergely Cs{\'a}ji, and Tam{\'a}s Kir{\'a}ly.
\newblock Manipulating the outcome of stable matching and roommates problems.
\newblock {\em arXiv preprint arXiv:2204.13485}, 2022.

\bibitem[\protect\citeauthoryear{Bil{\`o} \bgroup \em et al.\egroup
  }{2022}]{BMM2022FixedSizeHedon}
Vittorio Bil{\`o}, Gianpiero Monaco, and Luca Moscardelli.
\newblock Hedonic games with fixed-size coalitions.
\newblock In {\em Proceedings of the AAAI Conference on Artificial
  Intelligence}, volume~36, pages 9287--9295, 2022.

\bibitem[\protect\citeauthoryear{Bj{\"o}rklund \bgroup \em et al.\egroup
  }{2017}]{bjorklund2017narrow}
Andreas Bj{\"o}rklund, Thore Husfeldt, Petteri Kaski, and Mikko Koivisto.
\newblock Narrow sieves for parameterized paths and packings.
\newblock {\em Journal of Computer and System Sciences}, 87:119--139, 2017.

\bibitem[\protect\citeauthoryear{Bogomolnaia and
  Moulin}{2004}]{bogomolnaia2004random}
Anna Bogomolnaia and Herv{\'e} Moulin.
\newblock Random matching under dichotomous preferences.
\newblock {\em Econometrica}, 72(1):257--279, 2004.

\bibitem[\protect\citeauthoryear{Bogomolnaia and
  Moulin}{2023}]{bogomolnaia2023guarantees}
Anna Bogomolnaia and Herv{\'e} Moulin.
\newblock Guarantees in fair division: general or monotone preferences.
\newblock {\em Mathematics of Operations Research}, 48(1):160--176, 2023.

\bibitem[\protect\citeauthoryear{Bogomolnaia \bgroup \em et al.\egroup
  }{2005}]{bogomolnaia2005collective}
Anna Bogomolnaia, Herv{\'e} Moulin, and Richard Stong.
\newblock Collective choice under dichotomous preferences.
\newblock {\em Journal of Economic Theory}, 122(2):165--184, 2005.

\bibitem[\protect\citeauthoryear{Brandt \bgroup \em et al.\egroup
  }{2024}]{brandt2024coordinating}
Felix Brandt, Matthias Greger, Erel Segal-Halevi, and Warut Suksompong.
\newblock Coordinating charitable donations.
\newblock 2024.

\bibitem[\protect\citeauthoryear{Br{\^a}nzei \bgroup \em et al.\egroup
  }{2015}]{branzei2015characterization}
Simina Br{\^a}nzei, Hadi Hosseini, and Peter~Bro Miltersen.
\newblock Characterization and computation of equilibria for indivisible goods.
\newblock In {\em Algorithmic Game Theory: 8th International Symposium, SAGT
  2015}, pages 244--255. Springer, 2015.

\bibitem[\protect\citeauthoryear{Bredereck \bgroup \em et al.\egroup
  }{2020}]{BreCheFinNie2020-spscSM-jaamas}
Robert Bredereck, Jiehua Chen, Ugo~Paavo Finnendahl, and Rolf Niedermeier.
\newblock Stable roommates with narcissistic, single-peaked, and
  single-crossing preferences.
\newblock {\em Autonomous agents and multi-agent systems}, 34:1--29, 2020.

\bibitem[\protect\citeauthoryear{Caragiannis \bgroup \em et al.\egroup
  }{2021}]{caragiannis2019stable}
Ioannis Caragiannis, Aris Filos-Ratsikas, Panagiotis Kanellopoulos, and Rohit
  Vaish.
\newblock Stable fractional matchings.
\newblock {\em Artificial Intelligence}, 295:103416, 2021.

\bibitem[\protect\citeauthoryear{Cechlárová and
  Hajduková}{2004}]{CECHLAROVA2004333}
Katarína Cechlárová and Jana Hajduková.
\newblock Stable partitions with w-preferences.
\newblock {\em Discrete Applied Mathematics}, 138(3):333--347, 2004.

\bibitem[\protect\citeauthoryear{Chan \bgroup \em et al.\egroup
  }{2016}]{chan2016assignment}
Pak Chan, Xin Huang, Zhengyang Liu, Chihao Zhang, and Shengyu Zhang.
\newblock Assignment and pricing in roommate market.
\newblock In {\em Proceedings of the AAAI Conference on Artificial
  Intelligence}, volume~30, pages 446--452, 2016.

\bibitem[\protect\citeauthoryear{Chen and Roy}{2022}]{DBLP:conf/esa/0001R22}
Jiehua Chen and Sanjukta Roy.
\newblock Multi-dimensional stable roommates in 2-dimensional euclidean space.
\newblock In {\em 30th Annual European Symposium on Algorithms, {ESA}}, volume
  244, pages 36:1--36:16, 2022.

\bibitem[\protect\citeauthoryear{Chen \bgroup \em et al.\egroup
  }{2021}]{chen2021fractional}
Jiehua Chen, Sanjukta Roy, and Manuel Sorge.
\newblock Fractional matchings under preferences: Stability and optimality.
\newblock In {\em 30th International Joint Conference on Artificial
  Intelligence, IJCAI 2021}, pages 89--95, 2021.

\bibitem[\protect\citeauthoryear{Diamantoudi \bgroup \em et al.\egroup
  }{2004}]{diamantoudi2004random}
Effrosyni Diamantoudi, Eiichi Miyagawa, and Licun Xue.
\newblock Random paths to stability in the roommate problem.
\newblock {\em Games and Economic Behavior}, 48(1):18--28, 2004.

\bibitem[\protect\citeauthoryear{Dolev \bgroup \em et al.\egroup
  }{2012}]{dolev2012no}
Danny Dolev, Dror~G Feitelson, Joseph~Y Halpern, Raz Kupferman, and Nathan
  Linial.
\newblock No justified complaints: On fair sharing of multiple resources.
\newblock In {\em proceedings of the 3rd Innovations in Theoretical Computer
  Science Conference}, pages 68--75, 2012.

\bibitem[\protect\citeauthoryear{Dr{\'e}ze and Greenberg}{1980}]{DG1980Hedonic}
Jacques~H. Dr{\'e}ze and Joseph Greenberg.
\newblock Hedonic coalitions: {O}ptimality and stability.
\newblock {\em Econometrica}, 48(4):98--1003, 1980.

\bibitem[\protect\citeauthoryear{Elkind \bgroup \em et al.\egroup
  }{2021}]{elkind2021keeping}
Edith Elkind, Neel Patel, Alan Tsang, and Yair Zick.
\newblock Keeping your friends close: land allocation with friends.
\newblock In {\em Proceedings of the Twenty-Ninth International Conference on
  International Joint Conferences on Artificial Intelligence}, pages 318--324,
  2021.

\bibitem[\protect\citeauthoryear{Elkind \bgroup \em et al.\egroup
  }{2023}]{elkind2023justifying}
Edith Elkind, Piotr Faliszewski, Ayumi Igarashi, Pasin Manurangsi, Ulrike
  Schmidt-Kraepelin, and Warut Suksompong.
\newblock Justifying groups in multiwinner approval voting.
\newblock {\em Theoretical Computer Science}, 969:114039, 2023.

\bibitem[\protect\citeauthoryear{Feng \bgroup \em et al.\egroup
  }{2014}]{feng2014matching}
Qilong Feng, Jianxin Wang, and Jianer Chen.
\newblock Matching and weighted p2-packing: Algorithms and kernels.
\newblock {\em Theoretical Computer Science}, 522:85--94, 2014.

\bibitem[\protect\citeauthoryear{Filos-Ratsikas \bgroup \em et al.\egroup
  }{2014}]{filos2014social}
Aris Filos-Ratsikas, S{\o}ren Kristoffer~Stiil Frederiksen, and Jie Zhang.
\newblock Social welfare in one-sided matchings: Random priority and beyond.
\newblock In {\em International Symposium on Algorithmic Game Theory}, pages
  1--12. Springer, 2014.

\bibitem[\protect\citeauthoryear{Gale and Shapley}{1962}]{gale1962college}
David Gale and Lloyd~S Shapley.
\newblock College admissions and the stability of marriage.
\newblock {\em The American Mathematical Monthly}, 69(1):9--15, 1962.

\bibitem[\protect\citeauthoryear{Gan \bgroup \em et al.\egroup
  }{2023}]{gan2020fair}
Jiarui Gan, Bo~Li, and Yingkai Li.
\newblock Your college dorm and dormmates: Fair resource sharing with
  externalities.
\newblock {\em Journal of Artificial Intelligence Research}, 77:793--820, 2023.

\bibitem[\protect\citeauthoryear{Garey and Johnson}{1979}]{GJ79}
M.~R. Garey and D.~S. Johnson.
\newblock {\em Computers and Intractability: A Guide to the Theory of
  NP-Completeness}.
\newblock W. H. Freeman, 1979.

\bibitem[\protect\citeauthoryear{Ghodsi \bgroup \em et al.\egroup
  }{2011}]{ghodsi2011dominant}
Ali Ghodsi, Matei Zaharia, Benjamin Hindman, Andy Konwinski, Scott Shenker, and
  Ion Stoica.
\newblock Dominant resource fairness: fair allocation of multiple resource
  types.
\newblock In {\em Proceedings of the 8th USENIX conference on Networked Systems
  Design and Implementation}, pages 323--336, 2011.

\bibitem[\protect\citeauthoryear{Gibbard}{1973}]{gibbard1973manipulation}
Allan Gibbard.
\newblock Manipulation of voting schemes: a general result.
\newblock {\em Econometrica: Journal of the Econometric Society}, pages
  587--601, 1973.

\bibitem[\protect\citeauthoryear{Gokhale \bgroup \em et al.\egroup
  }{2024}]{gokhale2024capacity}
Salil Gokhale, Samarth Singla, Shivika Narang, and Rohit Vaish.
\newblock Capacity modification in the stable matching problem.
\newblock In {\em Proceedings of the 23rd International Conference on
  Autonomous Agents and Multiagent Systems}, AAMAS '24, page 697–705, 2024.

\bibitem[\protect\citeauthoryear{Green and Laffont}{1979}]{green1979incentives}
Jerry Green and Jean-Jacques Laffont.
\newblock {\em Incentives in public decision-making}.
\newblock Elsevier North-Holland, 1979.

\bibitem[\protect\citeauthoryear{Gudmundsson}{2014}]{gudmundsson2014stable}
Jens Gudmundsson.
\newblock When do stable roommate matchings exist? a review.
\newblock {\em Review of Economic Design}, 18:151--161, 2014.

\bibitem[\protect\citeauthoryear{Halpern \bgroup \em et al.\egroup
  }{2020}]{halpern2020fair}
Daniel Halpern, Ariel~D Procaccia, Alexandros Psomas, and Nisarg Shah.
\newblock Fair division with binary valuations: One rule to rule them all.
\newblock In {\em Web and Internet Economics: 16th International Conference,
  WINE 2020, Beijing, China, December 7--11, 2020, Proceedings 16}, pages
  370--383. Springer, 2020.

\bibitem[\protect\citeauthoryear{Hosseini \bgroup \em et al.\egroup
  }{2021}]{hosseini2021accomplice}
Hadi Hosseini, Fatima Umar, and Rohit Vaish.
\newblock Accomplice manipulation of the deferred acceptance algorithm.
\newblock In {\em Proceedings of the Thirtieth International Joint Conference
  on Artificial Intelligence, {IJCAI-21}}, pages 231--237, 2021.
\newblock Main Track.

\bibitem[\protect\citeauthoryear{Huzhang \bgroup \em et al.\egroup
  }{2017}]{huzhang2017online}
Guangda Huzhang, Xin Huang, Shengyu Zhang, and Xiaohui Bei.
\newblock Online roommate allocation problem.
\newblock In {\em Proceedings of the Twenty-Sixth International Joint
  Conference on Artificial Intelligence, {IJCAI-17}}, pages 235--241, 2017.

\bibitem[\protect\citeauthoryear{Hylland and
  Zeckhauser}{1979}]{hylland1979efficient}
Aanund Hylland and Richard Zeckhauser.
\newblock The efficient allocation of individuals to positions.
\newblock {\em Journal of Political economy}, 87(2):293--314, 1979.

\bibitem[\protect\citeauthoryear{Irving}{1985}]{IRVING1985577}
Robert~W Irving.
\newblock An efficient algorithm for the “stable roommates” problem.
\newblock {\em Journal of Algorithms}, 6(4):577--595, 1985.

\bibitem[\protect\citeauthoryear{Iwama \bgroup \em et al.\egroup
  }{2007}]{iwama2007stable}
Kazuo Iwama, Shuichi Miyazaki, and Kazuya Okamoto.
\newblock Stable roommates problem with triple rooms.
\newblock In {\em Proceedings of the 10th KOREA-JAPAN joint workshop on
  algorithms and computation (WAAC~'07)}, pages 105--112, 2007.

\bibitem[\protect\citeauthoryear{Kann}{1991}]{kann1991maximum}
Viggo Kann.
\newblock Maximum bounded 3-dimensional matching is max snp-complete.
\newblock {\em Information Processing Letters}, 37(1):27--35, 1991.

\bibitem[\protect\citeauthoryear{Knuth}{1976}]{knuth1976mariages}
Donald~Ervin Knuth.
\newblock {\em Mariages stables et leurs relations avec d'autres problemes
  combinatoires: introduction a l'analysis mathematique des algorithmes}.
\newblock Les Presses de l'Universit{\'e} de Montr{\'e}al, 1976.

\bibitem[\protect\citeauthoryear{Krysta \bgroup \em et al.\egroup
  }{2019}]{krysta2019size}
Piotr Krysta, David Manlove, Baharak Rastegari, and Jinshan Zhang.
\newblock Size versus truthfulness in the house allocation problem.
\newblock {\em Algorithmica}, 81:3422--3463, 2019.

\bibitem[\protect\citeauthoryear{Lam and Plaxton}{2019}]{lam2019existence}
Chi-Kit Lam and C~Gregory Plaxton.
\newblock On the existence of three-dimensional stable matchings with cyclic
  preferences.
\newblock In {\em International Symposium on Algorithmic Game Theory}, pages
  329--342. Springer, 2019.

\bibitem[\protect\citeauthoryear{Leontief}{1965}]{leontief1965structure}
Wassily~W Leontief.
\newblock The structure of the us economy.
\newblock {\em Scientific American}, 212(4):25--35, 1965.

\bibitem[\protect\citeauthoryear{Li and Xue}{2013}]{li2013egalitarian}
Jin Li and Jingyi Xue.
\newblock Egalitarian division under leontief preferences.
\newblock {\em Economic Theory}, 54(3):597--622, 2013.

\bibitem[\protect\citeauthoryear{Li \bgroup \em et al.\egroup
  }{2019}]{li2019room}
Yunpeng Li, Yichuan Jiang, Weiwei Wu, Jiuchuan Jiang, and Hui Fan.
\newblock Room allocation with capacity diversity and budget constraints.
\newblock {\em IEEE Access}, 7:42968--42986, 2019.

\bibitem[\protect\citeauthoryear{Liu \bgroup \em et al.\egroup
  }{2023}]{liu2023strategyproof}
Kwei-guu Liu, Kentaro Yahiro, and Makoto Yokoo.
\newblock Strategyproof mechanism for two-sided matching with resource
  allocation.
\newblock {\em Artificial Intelligence}, 316:103855, 2023.

\bibitem[\protect\citeauthoryear{Lu \bgroup \em et al.\egroup
  }{2024}]{lu2024approval}
Xinhang Lu, Jannik Peters, Haris Aziz, Xiaohui Bei, and Warut Suksompong.
\newblock Approval-based voting with mixed goods.
\newblock {\em Social Choice and Welfare}, pages 1--35, 2024.

\bibitem[\protect\citeauthoryear{Manlove}{2013}]{Manlove2013}
David~F. Manlove.
\newblock {\em Algorithmics of Matching Under Preferences}, volume~2 of {\em
  Series on Theoretical Computer Science}.
\newblock World Scientific, 2013.

\bibitem[\protect\citeauthoryear{McKay and Manlove}{2021}]{mckay2021three}
Michael McKay and David Manlove.
\newblock The three-dimensional stable roommates problem with additively
  separable preferences.
\newblock In {\em International Symposium on Algorithmic Game Theory}, pages
  266--280, 2021.

\bibitem[\protect\citeauthoryear{Narang and Narahari}{2020}]{narang2020study}
Shivika Narang and Yadati Narahari.
\newblock A study of incentive compatibility and stability issues in fractional
  matchings.
\newblock {\em 19th International Conference on Autonomous Agents and
  MultiAgent Systems}, pages 1951--1953, 2020.

\bibitem[\protect\citeauthoryear{Nicol{\'o} \bgroup \em et al.\egroup
  }{2019}]{NICOLO2019104942}
Antonio Nicol{\'o}, Arunava Sen, and Sonal Yadav.
\newblock Matching with partners and projects.
\newblock {\em Journal of Economic Theory}, 184:104942, 2019.

\bibitem[\protect\citeauthoryear{Nicol{\'o}}{2004}]{nicolo2004efficiency}
Antonio Nicol{\'o}.
\newblock Efficiency and truthfulness with leontief preferences. a note on
  two-agent, two-good economies.
\newblock {\em Review of Economic Design}, 8:373--382, 2004.

\bibitem[\protect\citeauthoryear{Nisan \bgroup \em et al.\egroup
  }{2007}]{nisan2007algorithmic}
Noam Nisan, Tim Roughgarden, Eva Tardos, and Vijay~V Vazirani.
\newblock Algorithmic game theory, 2007.
\newblock {\em Book available for free online}, 2007.

\bibitem[\protect\citeauthoryear{Ostrovsky and
  Rosenbaum}{2014}]{ostrovsky2014s}
Rafail Ostrovsky and Will Rosenbaum.
\newblock It's not easy being three: The approximability of three-dimensional
  stable matching problems.
\newblock {\em arXiv preprint arXiv:1412.1130}, 2014.

\bibitem[\protect\citeauthoryear{Parkes \bgroup \em et al.\egroup
  }{2015}]{parkes2015beyond}
David~C Parkes, Ariel~D Procaccia, and Nisarg Shah.
\newblock Beyond dominant resource fairness: Extensions, limitations, and
  indivisibilities.
\newblock {\em ACM Transactions on Economics and Computation (TEAC)},
  3(1):1--22, 2015.

\bibitem[\protect\citeauthoryear{Roth \bgroup \em et al.\egroup
  }{2005}]{roth2005pairwise}
Alvin~E Roth, Tayfun S{\"o}nmez, and M~Utku {\"U}nver.
\newblock Pairwise kidney exchange.
\newblock {\em Journal of Economic theory}, 125(2):151--188, 2005.

\bibitem[\protect\citeauthoryear{Roth}{1982}]{roth1982economics}
Alvin~E Roth.
\newblock The economics of matching: Stability and incentives.
\newblock {\em Mathematics of operations research}, 7(4):617--628, 1982.

\bibitem[\protect\citeauthoryear{Satterthwaite}{1975}]{satterthwaite1975strategy}
Mark~Allen Satterthwaite.
\newblock Strategy-proofness and arrow's conditions: Existence and
  correspondence theorems for voting procedures and social welfare functions.
\newblock {\em Journal of Economic Theory}, 10(2):187--217, 1975.

\bibitem[\protect\citeauthoryear{S{\"o}nmez}{1995}]{sonmez1995generalized}
Tayfun S{\"o}nmez.
\newblock Generalized matching problems.
\newblock {\em Ann Arbor}, 1001:48109--1220, 1995.

\bibitem[\protect\citeauthoryear{Svensson}{1999}]{svensson1999strategy}
Lars-Gunnar Svensson.
\newblock Strategy-proof allocation of indivisible goods.
\newblock {\em Social Choice and Welfare}, 16:557--567, 1999.

\bibitem[\protect\citeauthoryear{Teo and Sethuraman}{1998}]{teo1998geometry}
Chung-Piaw Teo and Jay Sethuraman.
\newblock The geometry of fractional stable matchings and its applications.
\newblock {\em Mathematics of Operations Research}, 23(4):874--891, 1998.

\bibitem[\protect\citeauthoryear{Teo \bgroup \em et al.\egroup
  }{2001}]{teo2001gale}
Chung-Piaw Teo, Jay Sethuraman, and Wee-Peng Tan.
\newblock Gale-shapley stable marriage problem revisited: Strategic issues and
  applications.
\newblock {\em Management Science}, 47(9):1252--1267, 2001.

\bibitem[\protect\citeauthoryear{Thiery and Ward}{2023}]{thiery2023improved}
Theophile Thiery and Justin Ward.
\newblock An improved approximation for maximum weighted k-set packing.
\newblock In {\em Proceedings of the 2023 Annual ACM-SIAM Symposium on Discrete
  Algorithms (SODA)}, pages 1138--1162. SIAM, 2023.

\bibitem[\protect\citeauthoryear{Wright and
  Vorobeychik}{2015}]{WV2015teamformation}
Mason Wright and Yevgeniy Vorobeychik.
\newblock Mechanism design for team formation.
\newblock In {\em Proceedings of the AAAI Conference on Artificial
  Intelligence}, volume~29, pages 1050--1056, 2015.

\bibitem[\protect\citeauthoryear{Zehavi}{2023}]{zehaviforgetfulness}
Meirav Zehavi.
\newblock Forgetfulness can make you faster: An o*(8.097 k)-time algorithm for
  weighted 3-set k-packing.
\newblock {\em ACM Transactions on Computation Theory}, 15(3-4):1--13, 2023.

\bibitem[\protect\citeauthoryear{Zhou}{1990}]{zhou1990conjecture}
Lin Zhou.
\newblock On a conjecture by gale about one-sided matching problems.
\newblock {\em Journal of Economic Theory}, 52(1):123--135, 1990.

\end{thebibliography}

\clearpage
\appendix
\section*{Appendix}

\section{Additional Related Work}\label{app:relwork}

A large variety of practical settings involve matching pairs or groups of agents. As a result, many different models of matchings have been studied over the years. These models differ in how many agents need to be matched, the type and format of the agents' preferences and the objectives being pursued. 

\paragraph{Matching Roommates to Rooms.} There has been some prior work on the same model of matching roommates to rooms with agents having preferences over both their roommate as well as their assigned room. The model was first studied by \citet{chan2016assignment}. They consider the problem of maximizing social welfare in this model under additive utilities. They show that this problem is NP-hard and give an algorithm that is \USW{$\sfrac{2}{3}$}. They also study two extensions of exchange stability for this setting called two-person stability (2PS) and four-person stability (4PS). Both these notions are different from the standard notion of stability. The paper also considers fairness where different notions of fairness are achieved by levying an appropriately high rent on agents being envied by other agents. 

\citet{gan2020fair} also consider fairness in roommate matchings for rooms of size 2 or more, however they do not study a model with rents. They show that for arbitrary room sizes, finding an envy-free roommate matching is NP-hard. They then propose a relaxation of envy-freeness called {\em Pareto envy-freeness}. They give an algorithm to find a pareto envy-free roommate matching for rooms of size 2 under binary and symmetric agent valuations. Other work on this model are extensions of the work of \citet{chan2016assignment}. \citet{huzhang2017online} study an online variant of the problem and give a linear competitive ratio. Meanwhile, \citet{li2019room} consider an extension with rooms of different sizes and each agent having a maximum budget for rent. None of these papers consider strategyproofness or Leontief utilities. 

Nicol{\'o} et al.~\cite{NICOLO2019104942} studied stategyproofness for core in a similar setting where pairs of agents are assigned to projects.

\paragraph{Land Allocation.}
A somewhat similar model looks at land allocation. Here, agents wish to buy plots of land, and have preferences over the different plots, but also have preferences over their neighbours. \citet{elkind2021keeping} look at the problem of maximizing social welfare, which turns out to NP-hard. They also look for strategy proof mechanisms which give approximations on the social welfare. \citet{aziz2021individually} follows up on this work showing considering strategyproofness along with other efficiency notions like pareto optimality and individual rationality. This model is very different from the standard models on matchings. 

\paragraph{House Allocation.} 
Also often called object allocation or the assignment problem, this is the simplest matching model where each agent receives one object/house.
Here, agents must be matched to houses and often, they come with a initial endowment. 
This model has been widely studied, with  a variety of objectives like fairness and economic efficiency through the lens of pareto optimality  and social welfare \citep{hylland1979efficient,abraham2004pareto,filos2014social,krysta2019size}.

Like our study, a significant amount of work looks at the coexistence of strategyproofness and economic efficiency \citep{svensson1999strategy,abdulkadirouglu2003school}. While their incompatibility is well known \citep{green1979incentives}, a seminal result by \citet{hylland1979efficient} overcomes this using ``payments" in the form of probabilities of assignment. \citet{krysta2019size} adapt the standard serial dictatorship mechanism  to be able to achieve approximate \USW{} and strategyproofness. 

\paragraph{Two-sided and Roommate Matchings.} The large majority of work on matchings looks at matching pairs of agents. A standard assumption is that these matchings are {\em two-sided}, i.e. agents are partitioned into two sets, and each pair in the matching must contain one agent from each set. The seminal paper by \citet{gale1962college} also considers such a setting. Typical work on two-sided matchings assumes ordinal preferences over agents, and as a result, social welfare considerations are only in the papers that consider some form of cardinal valuations \citep{bogomolnaia2004random,caragiannis2019stable,chen2021fractional,narang2020study}.

Strategyproofness and incentives to manipulate matching mechanisms have been well studied, beginning with the result of \citet{roth1982economics}. \citet{roth1982economics} showed that no stable matching mechanism can be strategyproof for all agents, but the deferred acceptance mechanism given by \citet{gale1962college} is strategyproof for the proposing side (but not hte accepting side). Following this work, a body of work has emerged, exploring the manipulation of deferred acceptance mechanism \citep{teo2001gale, hosseini2021accomplice, gokhale2024capacity}. The work on strategyproof matching mechanisms extends well beyond the deferred acceptance mechanism however \citep{narang2020study,gokhale2024capacity}. 

There is a significant amount of work that looks at what has traditionally been called roommate matchings. These are simply matching pairs of agents with no partition over the agents restricting the matches, or any explicit rooms or preferences over them. Quite a bit of work in this area has focused on stable matchings \citep{IRVING1985577,DBLP:conf/esa/0001R22,diamantoudi2004random,BreCheFinNie2020-spscSM-jaamas, gudmundsson2014stable}. Stable roommate matchings need not exist so much of this work has gone finding one when it exists \citep{IRVING1985577,gudmundsson2014stable} and characterizing the set of properties that ensure existence \citep{DBLP:conf/esa/0001R22,BreCheFinNie2020-spscSM-jaamas,gudmundsson2014stable}.

Strategyproof mechanisms have been studied often for roommate matchings \cite{sonmez1995generalized,berczi2022manipulating}.
A related model is of pairwise kidney exchange where strategyproofness has been well studied as well \citep{roth2005pairwise, ashlagi2015mix}. 

\paragraph{Three Dimensional Matchings and Hedonic Games.}

The work on roommate matchings typically looks at rooms of size two, that is, agents must be paired up. However, some work does look at rooms of larger sizes, known as multi-dimensional matching or fixed-sized coalition formation~\citep{knuth1976mariages,DG1980Hedonic}. When the room-size is at least three, the problem of finding a stable matching becomes NP-hard~\citep{iwama2007stable}. \cite{ostrovsky2014s} considered the question of finding an approximate solution for the case when room-size is three. Multi-dimensional matchings have been studied on restricted preference domains as well~\citep{lam2019existence,BreCheFinNie2020-spscSM-jaamas,DBLP:conf/esa/0001R22}. We refer to the textbook by \cite{Manlove2013} for a detailed overview. 

Hedonic games with fixed-size coalitions have been studied for other solution concepts such as strategyproofness~\citep{WV2015teamformation}, Pareto optimality~\citep{aziz2013pareto}, and exchange stability~\citep{BMM2022FixedSizeHedon}. 
The focus of the models in these areas is on preferences over agents and no preferences over rooms are given.

\paragraph{Leontief Utilities.}

Leontief utilities have been well studied in the EconCS community for a variety of models and objectives. They capture one extreme of aggregate preferences which complement a substitutionary approach modelled by additive utilities. 
\citet{nicolo2004efficiency} shows that  leontief utilities enable the existence of efficient and strategyproof mechanisms in exchange economies., which need not exist for arbitrary utility functions. A large body of work has considered Leontief utilities in context of  Nash equilibria and auctions. See \citet{nisan2007algorithmic} for a detailed overview.  

Leontief utilities help capture a Rawlsian view of fairness. To this end, they have been considered in the context of fair division \citep{li2013egalitarian, bogomolnaia2023guarantees} and fair resource allocation \citep{ghodsi2011dominant,parkes2015beyond, branzei2015characterization, dolev2012no}. In the setting of fixed-size coalitions and hedonic games, Leontief utilities have been studied under the name of W-preferences~\cite{CECHLAROVA2004333,aziz2012individual,aziz2013pareto}.





\section{Omitted Material from Section 3}\label{app:maxwelfare}
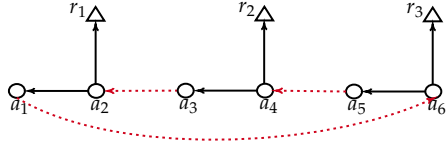
\begin{figure}[t]
    \centering
\tikzset{every picture/.style={line width=0.75pt}} 

\begin{tikzpicture}[x=0.75pt,y=0.75pt,yscale=-1,xscale=1]

\draw   (49.98,4.24) -- (54.41,11.21) -- (45.55,11.21) -- cycle ;
\draw [color={rgb, 255:red, 208; green, 2; blue, 27 }  ,draw opacity=1 ] [dash pattern={on 1pt off 1.5pt}]  (91.4,46.48) -- (56.32,46.63) ;
\draw [shift={(56,46.64)}, rotate = 359.76] [color={rgb, 255:red, 208; green, 2; blue, 27 }  ,draw opacity=1 ][line width=0.75]   (3,-1.5) .. controls (3,-0.7) and (1.7,-0.3) .. (0,0) .. controls (1.7,0.3) and (3,0.7) .. (3,1.5)  ;
\draw   (46.26,46.64) .. controls (46.26,48.56) and (48.07,50.12) .. (50.29,50.12) .. controls (52.51,50.12) and (54.32,48.56) .. (54.32,46.64) .. controls (54.32,44.71) and (52.51,43.15) .. (50.29,43.15) .. controls (48.07,43.15) and (46.26,44.71) .. (46.26,46.64) -- cycle ;
\draw   (6.26,46.64) .. controls (6.26,48.56) and (8.07,50.12) .. (10.29,50.12) .. controls (12.51,50.12) and (14.32,48.56) .. (14.32,46.64) .. controls (14.32,44.71) and (12.51,43.15) .. (10.29,43.15) .. controls (8.07,43.15) and (6.26,44.71) .. (6.26,46.64) -- cycle ;
\draw    (46.26,46.64) -- (16.32,46.64) ;
\draw [shift={(16,46.64)}, rotate = 360] [color={rgb, 255:red, 0; green, 0; blue, 0 }  ][line width=0.75]    (3,-1.5) .. controls (3,-0.7) and (1.7,-0.3) .. (0,0) .. controls (1.7,0.3) and (3,0.7) .. (3,1.5)  ;
\draw    (50,43.15) -- (50,12.67) ;
\draw [shift={(50,13)}, rotate = 89.88] [color={rgb, 255:red, 0; green, 0; blue, 0 }  ][line width=0.75]   (3,-1.5) .. controls (3,-0.7) and (1.7,-0.3) .. (0,0) .. controls (1.7,0.3) and (3,0.7) .. (3,1.5)  ;
\draw   (135.11,4.09) -- (139.54,11.06) -- (130.68,11.06) -- cycle ;
\draw   (131.4,46.48) .. controls (131.4,48.41) and (133.2,49.97) .. (135.42,49.97) .. controls (137.65,49.97) and (139.45,48.41) .. (139.45,46.48) .. controls (139.45,44.56) and (137.65,43) .. (135.42,43) .. controls (133.2,43) and (131.4,44.56) .. (131.4,46.48) -- cycle ;
\draw   (91.4,46.48) .. controls (91.4,48.41) and (93.2,49.97) .. (95.42,49.97) .. controls (97.65,49.97) and (99.45,48.41) .. (99.45,46.48) .. controls (99.45,44.56) and (97.65,43) .. (95.42,43) .. controls (93.2,43) and (91.4,44.56) .. (91.4,46.48) -- cycle ;
\draw    (131.4,46.48) -- (101.45,46.48) ;
\draw [shift={(101.5,46.48)}, rotate = 360] [color={rgb, 255:red, 0; green, 0; blue, 0 }  ][line width=0.75]    (3,-1.5) .. controls (3,-0.7) and (1.7,-0.3) .. (0,0) .. controls (1.7,0.3) and (3,0.7) .. (3,1.5)  ;
\draw    (135,43) -- (135,12.51) ;
\draw [shift={(135,12.5)}, rotate = 89.88] [color={rgb, 255:red, 0; green, 0; blue, 0 }  ][line width=0.75]    (3,-1.5) .. controls (3,-0.7) and (1.7,-0.3) .. (0,0) .. controls (1.7,0.3) and (3,0.7) .. (3,1.5)  ;
\draw   (220,4.53) -- (224.43,11.5) -- (215.57,11.5) -- cycle ;
\draw   (216.29,46.93) .. controls (216.29,48.85) and (218.09,50.41) .. (220.31,50.41) .. controls (222.54,50.41) and (224.34,48.85) .. (224.34,46.93) .. controls (224.34,45) and (222.54,43.44) .. (220.31,43.44) .. controls (218.09,43.44) and (216.29,45) .. (216.29,46.93) -- cycle ;
\draw   (176.29,46.93) .. controls (176.29,48.85) and (178.09,50.41) .. (180.31,50.41) .. controls (182.54,50.41) and (184.34,48.85) .. (184.34,46.93) .. controls (184.34,45) and (182.54,43.44) .. (180.31,43.44) .. controls (178.09,43.44) and (176.29,45) .. (176.29,46.93) -- cycle ;
\draw    (216.29,46.93) -- (186.34,46.93) ;
\draw [shift={(186,46.93)}, rotate = 360] [color={rgb, 255:red, 0; green, 0; blue, 0 }  ][line width=0.75]     (3,-1.5) .. controls (3,-0.7) and (1.7,-0.3) .. (0,0) .. controls (1.7,0.3) and (3,0.7) .. (3,1.5)  ;
\draw    (220,43.44) -- (220,12.96) ;
\draw [shift={(220,13)}, rotate = 89.88] [color={rgb, 255:red, 0; green, 0; blue, 0 }  ][line width=0.75]    (3,-1.5) .. controls (3,-0.7) and (1.7,-0.3) .. (0,0) .. controls (1.7,0.3) and (3,0.7) .. (3,1.5)  ;
\draw [color={rgb, 255:red, 208; green, 2; blue, 27 }  ,draw opacity=1 ] [dash pattern={on 1pt off 1.5pt}]  (176.29,46.93) -- (141.45,46.51) ;
\draw [shift={(141,46.48)}, rotate = 0.69] [color={rgb, 255:red, 208; green, 2; blue, 27 }  ,draw opacity=1 ][line width=0.75]   (3,-1.5) .. controls (3,-0.7) and (1.7,-0.3) .. (0,0) .. controls (1.7,0.3) and (3,0.7) .. (3,1.5)  ;
\draw [color={rgb, 255:red, 208; green, 2; blue, 27 }  ,draw opacity=1 ] [dash pattern={on 1pt off 1.5pt}]  (10.29,50.12) .. controls (59.3,79.9) and (163.69,76.28) .. (218.66,51.18) ;
\draw [shift={(220.31,50.41)}, rotate = 154.68] [color={rgb, 255:red, 208; green, 2; blue, 27 }  ,draw opacity=1 ][line width=0.75]   (3,-1.5) .. controls (3,-0.7) and (1.7,-0.3) .. (0,0) .. controls (1.7,0.3) and (3,0.7) .. (3,1.5)  ;;

\draw (5.53,49.44) node [anchor=north west][inner sep=0.75pt]  [font=\scriptsize]  {$a_{1} \ $};
\draw (45.98,49.4) node [anchor=north west][inner sep=0.75pt]  [font=\scriptsize]  {$a_{2}$};
\draw (35.33,0.53) node [anchor=north west][inner sep=0.75pt]  [font=\scriptsize]  {$r_{1}$};
\draw (90.67,49.29) node [anchor=north west][inner sep=0.75pt]  [font=\scriptsize]  {$a_{3} \ $};
\draw (131.12,49.25) node [anchor=north west][inner sep=0.75pt]  [font=\scriptsize]  {$a_{4}$};
\draw (120.47,0.15) node [anchor=north west][inner sep=0.75pt]  [font=\scriptsize]  {$r_{2}$};
\draw (175.56,49.73) node [anchor=north west][inner sep=0.75pt]  [font=\scriptsize]  {$a_{5} \ $};
\draw (216.01,49.69) node [anchor=north west][inner sep=0.75pt]  [font=\scriptsize]  {$a_{6}$};
\draw (205.36,0.82) node [anchor=north west][inner sep=0.75pt]  [font=\scriptsize]  {$r_{3}$};

\end{tikzpicture}

    \caption{An instance showing an example where the \USW{$\sfrac{2}{3}$ algorithm given by \citet{chan2016assignment} (for additive utilities) is \USW{0} under Leontief utilities.}}
    \label{fig:ChanCounter}
\end{figure}

In this section, we provide the proofs omitted from \cref{sec:swapprox}. %

\paragraph{Counter Example.} We first show an example in \Cref{fig:ChanCounter} where the algorithm by \citet{chan2016assignment} which is \USW{$\sfrac{2}{3}$} for additive utilities is \USW{0} for Leontief. Here an optimal matching under both Leontief and additive utilities is shown by black solid lines, where agents $a_2$, $a_4$ and $a_6$ get utility $1$ each and the other agents get utility $0$. The algorithm by \citet{chan2016assignment} proceeds by finding an agent-agent maximum weight matching and the an agent-room max weight matching and combining them. Here one possible agent-agent max weight matching is shown by the red dashed edges. 

There is only one agent-room max weight matching shown by solid edges. When these two are combined, it produces a matching which is optimal for additive utilities, but gives all agents utility $0$ under Leontief. As a result, the algorithm by \citet{chan2016assignment} cannot be used for the Leontief setting. 

We now show provide the omitted proofs for our intractability results.

\subsection{Intractability Results}\label{app:red}

We now prove the inapproximability of maximizing social welfare under binary symmetric Leontief utilities. This proof is similar to that of $3$DM. In fact, the ideas we present in this reduction can be applied to the reduction by \citet{kann1991maximum} to show that this problem is in fact MaxSNP-hard. 

\swhardl*
\begin{proof}
Given a Max3SAT instance, i.e., a $3$-CNF formula with $\enn$ variables $\{x_1, x_2, \dots, x_{\enn}\}$ and $\emm$ clauses $\{c_1, c_2, \dots, c_{\emm}\}$, we  construct an instance of {\sc Roommate Matching} with Leontief utilities where valuations are binary and symmetric, and each matched triple can give social welfare at most $1$.  For each clause $c_j$, we shall create a simple gadget with one agent and one room. For each variable $x_i$, we shall create a variable gadget with $3d_i$ agents and $d_i$ rooms, where $d_i$ is the number of clauses $x_i$ appears in.\\

\noindent {\bf Clause gadget.} For each clause $c_j$, $j\in [\emm]$, create a clause agent $c_j$ and a clause room $r_j$. Agent $c_j$ has value $1$ for room $r_j$. \\

\noindent {\bf Variable gadget.} For each variable $x_i$, $i\in [\enn]$, we create a total of $3d_i$ agents of two types, namely, $a_i$-agents, $a_{i1}, a_{i2},\dots, a_{id_i}$, and $b_i$-agents, $b_{i1},b_{i2},\dots, b_{i2d_i}$, and we create $d_i$ variable-rooms $r_{i1},\dots, r_{id_i}$. Of these $b_i$ agents, agents $b_{i1}, b_{i3}, \dots, b_{i2d_i-1}$ are called odd agents and agents $b_{i2}, b_{i4}, \dots, b_{i2d_i}$ are called even agents. We set the preferences inside the variable gadget as follows: for each $t \in {d_i}$,  $\widehat{v}_{b_{i2t-1}}(r_{it})=\widehat{v}_{b_{i2t}}(r_{it})=1$. Further, for each $t\in [d_i -1]$ $v_{b_{i2t}}(a_{it})=v_{b_{i2t+1}}(a_{it})=1$, with $v_{b_{i2d_i}}(a_{id_i})=v_{b_{i1}}(a_{id_i})=1$.\\

\noindent {\bf Connection between clause and variable gadgets.}  Agents $b_{i2t-1}$ and $b_{i2t}$ correspond to the occurrence of variable $x_i$ in clause $c_t$. If $x_i$ appears as a positive literal in $c_t$, then $b_{i2t-1}$ will have value for the clause agent $c_t$, if it appears as a negative literal then $b_{i2t}$ will have value $1$ for the clause agent $c_t$. Else, if it does not appear in clause $c_t$, then all agents in the variable gadget have value $0$ for the clause agent and room.

All other values are $0$. In particular, no agent in the variable gadget has value for a clause room.  An example of such a construction is given in Figure \ref{fig:lsred} for the formula $(x_1 \vee \Bar{x_2})(\Bar{x_1}\vee x_2)$.   All valuations are symmetric.  This ends our construction. We show that a given instance of Max$3$SAT has at least $k$ clauses simultaneously satisfiable if and only if there exists a matching with social welfare at least $3\emm + k$ in the constructed {\sc Roommate Matching} instance. 

Suppose there exists a matching $\mu$ with social welfare at least $3\emm + k$. We begin by observing some properties of a social welfare maximizing matching. By the construction, there is no two agents in the instance who like each other and like the same room. Therefore, each triple can give social welfare at most one. Therefore, since there are $\emm +\sum_{i}d_i=4\emm$ rooms, each matched triple adds one to the social welfare. In particular, each variable gadget adds $\emm$ to the social welfare and each clause gadget adds one.
    
We analyse the triples in a variable gadget and show that there are two possible matchings such that each triple give social welfare at least one. Recall that $\USW{}(\mu)\geq (\enn+1)\emm$.
    
    \begin{observation}\label{obs:nphLs-odd-even}
    In the variable gadget of $x_i$, either all odd $b_i$-agents $b_{i1},b_{i3},\dots, b_{i2d_i-1}$ have utility one or the even $b_i$-agents $b_{i2},b_{i4},\dots, b_{i2d_i}$ have utility one in $\mu$.
    \end{observation}
    \begin{proof}
    Note that any agent of the type $a_{i}$ can never get a value of $1$, as it has value zero for each room for each $i \in [\enn]$. Thus, only the clause agents $c_j$ and each agent of the type $b_{is}$ can get a value of $1$, with $i\in [\enn],\, j\in [\emm]$ and $s\in [2d_i]$. 
    If an $b_i$-agent is not matched to a room within the variable gadget of $x_i$, then it cannot get a value of $1$. Thus, as the number of such rooms are half the number of $b_i$-agents, and no two $b_i$-agents like each other, at most half of the agents among $b_{i1},\dots, b_{i2d_i}$ can get value $1$. These would be either the odd $b_i$-agents $b_{i1},b_{i3},\dots, b_{i2d_i-1}$ or the even $b_i$-agents $b_{i2},b_{i4},\dots, b_{i2d_i}$. Either of these two would give a social welfare of $d_i$. 
    \end{proof}
    
    In the example in \cref{fig:lsred}, either the orange edges can be matched or the blue ones can be matched but a blue and an orange edge cannot be matched simultaneously to give maximum social welfare. We demonstrate this in \cref{fig:lsred}. The edges in blue and yellow show these two possible solutions.  Therefore, the total utility of the variable-agents is at most $d_i$ in each variable gadget, giving a total of at most $\sum_i d_i=3\emm$ from the $\enn$ variable gadgets. Hence, each clause agent must have utility one. 

    Next, we analyse the matching of the clause agents. 
    Recall that each clause agent $c_j$ gets value $1$ if and only if it is matched to $r_j$ and an agent of the type $b_{i2j-1}$ or $b_{i2j}$.
    We show the clause agents are matched consistently. 
    \begin{observation}\label{obs:nphLs-clause-odd-even}
        Let $x_i$ be a variable that appears in two clauses $c_t$ and $c_{t'}$. If the two clause agents $c_t$ and $c_{t'}$ are matched to $b_i$-agents, then either both $c_t$ and $c_{t'}$, are matched to odd $b_i$-agents, or both are matched to even $b_i$ agents.
    \end{observation}
    \begin{proof}
     If agents $b_{i2t-1}$ and $b_{i2t'}$ are matched to clause agents $c_t$ and $c_{t'}$, respectively, then, there is at least one odd $b_i$-agent and one even $b_i$-agent who have utility zero. Then, using \Cref{obs:nphLs-odd-even}, the social welfare generated from the matchings within the variable gadget of $x_i$ is less than $d_i$, contradicting the fact that the social welfare is $3\emm+k$.
    \end{proof}

    We demonstrate this in \cref{fig:lsred} by edges in green show clause gadget edges that are always matched in an optimal solution. 
    
    Given an optimal roommate matching with social welfare $3\emm +k$, we construct an assignment satisfying $k$ clauses for the $3$SAT instance as follows. For each variable $x_i$, if an odd $b_i$-agent is matched to a clause agent, then set $x_i$ to true. If an even $b_i$ agent is matched to a clause agent, set $x_i$ to false. From \Cref{obs:nphLs-clause-odd-even}, it directly follows that the assignment is well defined. Since each clause is matched to some $b_i$-agent, then each clause is satisfied in the above assignment.
    
    
    
    Analogously, when the given formula has $k$ clauses satisfiable, we can construct a matching where each of the clause agents (of the clauses satisfied) gets value $1$ by matching the clause agent $c_j$ to the clause room $r_j$ and to $b_i$-agent that correspond to the true literal that satisfy the clause. If variable $x_i$ is true (or false), then within the variable gadget, the even $b_i$s (resp. odd $b_i$s) are matched to variable rooms and the $a_i$s they like. The matching gives social welfare $\emm$ within each variable gadget. Hence, the social welfare is $(\enn+1)\emm$.  Thus, the given $3$SAT formula has $k$ clauses satisfiable if and only if the constructed instance has a matching which gives social welfare $3\emm+k$.

    \paragraph{Approximation Preservation.} We now show that the optimal welfare of the constructed instance is upper bounded by a constant multiple of the maximum number of clauses satisfiable. For a given $3$SAT formula, at least half the clauses can always be satisfied. We use this to upper bound the maximum possible social welfare. Let $k^*$ be the maximum number of satisfiable clauses for the given $3$SAT formula.As at least half the clauses are always satisfiable, $\emm/2\leq k^*$. The maximum possible social welfare of the constructed instance is:

    \begin{align*}
        \sum_{i\in [\enn]}d_i + k^* &= 3\emm +k^*\\
                                    &\leq 7k^*
    \end{align*}    

    Hence, a PTAS for maximizing social welfare would imply a a PTAS for solving Max$3$SAT.    
\end{proof}

\begin{figure}[t]
    \centering
   
\tikzset{every picture/.style={line width=0.85pt}} 

\begin{tikzpicture}[x=0.75pt,y=0.75pt,yscale=-1,xscale=1]

\draw   (16.94,34) .. controls (16.94,31.4) and (19.1,29.28) .. (21.78,29.28) .. controls (24.45,29.28) and (26.62,31.4) .. (26.62,34) .. controls (26.62,36.61) and (24.45,38.72) .. (21.78,38.72) .. controls (19.1,38.72) and (16.94,36.61) .. (16.94,34) -- cycle ;
\draw   (76.76,93.48) .. controls (76.76,90.88) and (78.93,88.77) .. (81.6,88.77) .. controls (84.27,88.77) and (86.44,90.88) .. (86.44,93.48) .. controls (86.44,96.09) and (84.27,98.2) .. (81.6,98.2) .. controls (78.93,98.2) and (76.76,96.09) .. (76.76,93.48) -- cycle ;
\draw   (76.24,63.99) .. controls (76.24,61.38) and (78.41,59.27) .. (81.08,59.27) .. controls (83.76,59.27) and (85.92,61.38) .. (85.92,63.99) .. controls (85.92,66.6) and (83.76,68.71) .. (81.08,68.71) .. controls (78.41,68.71) and (76.24,66.6) .. (76.24,63.99) -- cycle ;
\draw   (76.24,34) .. controls (76.24,31.4) and (78.41,29.28) .. (81.08,29.28) .. controls (83.76,29.28) and (85.92,31.4) .. (85.92,34) .. controls (85.92,36.61) and (83.76,38.72) .. (81.08,38.72) .. controls (78.41,38.72) and (76.24,36.61) .. (76.24,34) -- cycle ;
\draw   (17.2,93.48) .. controls (17.2,90.88) and (19.36,88.77) .. (22.04,88.77) .. controls (24.71,88.77) and (26.88,90.88) .. (26.88,93.48) .. controls (26.88,96.09) and (24.71,98.2) .. (22.04,98.2) .. controls (19.36,98.2) and (17.2,96.09) .. (17.2,93.48) -- cycle ;
\draw   (16.94,63.99) .. controls (16.94,61.38) and (19.1,59.27) .. (21.78,59.27) .. controls (24.45,59.27) and (26.62,61.38) .. (26.62,63.99) .. controls (26.62,66.6) and (24.45,68.71) .. (21.78,68.71) .. controls (19.1,68.71) and (16.94,66.6) .. (16.94,63.99) -- cycle ;
\draw   (50.84,30.27) -- (55.44,38.36) -- (46.23,38.36) -- cycle ;
\draw   (50.84,89.24) -- (55.44,97.33) -- (46.23,97.33) -- cycle ;
\draw [color={rgb, 255:red, 245; green, 166; blue, 35 }  ,draw opacity=1 ]   (21.78,38.72) -- (21.78,59.27) ;
\draw [color={rgb, 255:red, 245; green, 166; blue, 35 }  ,draw opacity=1 ]   (85.92,34) -- (107.95,34) ;
\draw [color={rgb, 255:red, 245; green, 166; blue, 35 }  ,draw opacity=1 ]   (27,34) -- (48,34) ;
\draw [color={rgb, 255:red, 74; green, 144; blue, 226 }  ,draw opacity=1 ]   (76.24,34) -- (54,34) ;
\draw [color={rgb, 255:red, 74; green, 144; blue, 226 }  ,draw opacity=1 ]   (81,38.72) -- (81,59.27) ;
\draw [color={rgb, 255:red, 245; green, 166; blue, 35 }  ,draw opacity=1 ]   (81,68.71) -- (81,88.77) ;
\draw [color={rgb, 255:red, 245; green, 166; blue, 35 }  ,draw opacity=1 ]   (76.3,93) -- (54,93) ;
\draw [color={rgb, 255:red, 74; green, 144; blue, 226 }  ,draw opacity=1 ]   (26.88,93) -- (48,93) ;
\draw [color={rgb, 255:red, 74; green, 144; blue, 226 }  ,draw opacity=1 ]   (22,68.71) -- (22,88.77) ;
\draw   (107.95,34) .. controls (107.95,31.4) and (110.12,29.28) .. (112.79,29.28) .. controls (115.46,29.28) and (117.63,31.4) .. (117.63,34) .. controls (117.63,36.61) and (115.46,38.72) .. (112.79,38.72) .. controls (110.12,38.72) and (107.95,36.61) .. (107.95,34) -- cycle ;
\draw   (112.57,60) -- (117.17,68.09) -- (107.96,68.09) -- cycle ;
\draw [color={rgb, 255:red, 126; green, 211; blue, 33 }  ,draw opacity=1 ]   (112.57,39) -- (112.57,60) ;
\draw [color={rgb, 255:red, 74; green, 144; blue, 226 }  ,draw opacity=1 ]   (117.63,34) -- (167.52,34) ;
\draw   (167.52,34) .. controls (167.52,31.4) and (169.69,29.28) .. (172.36,29.28) .. controls (175.04,29.28) and (177.2,31.4) .. (177.2,34) .. controls (177.2,36.61) and (175.04,38.72) .. (172.36,38.72) .. controls (169.69,38.72) and (167.52,36.61) .. (167.52,34) -- cycle ;
\draw   (229.35,93.48) .. controls (229.35,90.88) and (231.52,88.77) .. (234.19,88.77) .. controls (236.86,88.77) and (239.03,90.88) .. (239.03,93.48) .. controls (239.03,96.09) and (236.86,98.2) .. (234.19,98.2) .. controls (231.52,98.2) and (229.35,96.09) .. (229.35,93.48) -- cycle ;
\draw   (228.83,63.99) .. controls (228.83,61.38) and (231,59.27) .. (233.67,59.27) .. controls (236.34,59.27) and (238.51,61.38) .. (238.51,63.99) .. controls (238.51,66.6) and (236.34,68.71) .. (233.67,68.71) .. controls (231,68.71) and (228.83,66.6) .. (228.83,63.99) -- cycle ;
\draw   (228.83,34) .. controls (228.83,31.4) and (231,29.28) .. (233.67,29.28) .. controls (236.34,29.28) and (238.51,31.4) .. (238.51,34) .. controls (238.51,36.61) and (236.34,38.72) .. (233.67,38.72) .. controls (231,38.72) and (228.83,36.61) .. (228.83,34) -- cycle ;
\draw   (167.78,93.48) .. controls (167.78,90.88) and (169.95,88.77) .. (172.62,88.77) .. controls (175.3,88.77) and (177.46,90.88) .. (177.46,93.48) .. controls (177.46,96.09) and (175.3,98.2) .. (172.62,98.2) .. controls (169.95,98.2) and (167.78,96.09) .. (167.78,93.48) -- cycle ;
\draw   (167.52,63.99) .. controls (167.52,61.38) and (169.69,59.27) .. (172.36,59.27) .. controls (175.04,59.27) and (177.2,61.38) .. (177.2,63.99) .. controls (177.2,66.6) and (175.04,68.71) .. (172.36,68.71) .. controls (169.69,68.71) and (167.52,66.6) .. (167.52,63.99) -- cycle ;
\draw   (202.42,30.27) -- (207.03,38.36) -- (197.82,38.36) -- cycle ;
\draw   (202.42,89.24) -- (207.03,97.33) -- (197.82,97.33) -- cycle ;
\draw [color={rgb, 255:red, 245; green, 166; blue, 35 }  ,draw opacity=1 ]   (172.36,38.72) -- (172.36,59.27) ;
\draw [color={rgb, 255:red, 245; green, 166; blue, 35 }  ,draw opacity=1 ]   (167.78,93.48) -- (146.95,93.48) ;
\draw [color={rgb, 255:red, 245; green, 166; blue, 35 }  ,draw opacity=1 ]   (177.2,34) -- (199.42,34) ;
\draw [color={rgb, 255:red, 74; green, 144; blue, 226 }  ,draw opacity=1 ]   (228.83,34) -- (205,34) ;
\draw [color={rgb, 255:red, 74; green, 144; blue, 226 }  ,draw opacity=1 ]   (233.67,38.72) -- (233.67,59.27) ;
\draw [color={rgb, 255:red, 245; green, 166; blue, 35 }  ,draw opacity=1 ]   (233.67,68.71) -- (233.67,88.77) ;
\draw [color={rgb, 255:red, 245; green, 166; blue, 35 }  ,draw opacity=1 ]   (229.35,93) -- (205.8,93) ;
\draw [color={rgb, 255:red, 74; green, 144; blue, 226 }  ,draw opacity=1 ]   (177.46,93) -- (199.42,93) ;
\draw [color={rgb, 255:red, 74; green, 144; blue, 226 }  ,draw opacity=1 ]   (172.36,68.71) -- (172.62,88.77) ;
\draw   (137.27,93.48) .. controls (137.27,90.88) and (139.43,88.77) .. (142.11,88.77) .. controls (144.78,88.77) and (146.95,90.88) .. (146.95,93.48) .. controls (146.95,96.09) and (144.78,98.2) .. (142.11,98.2) .. controls (139.43,98.2) and (137.27,96.09) .. (137.27,93.48) -- cycle ;
\draw   (141.98,60) -- (146.58,68.09) -- (137.37,68.09) -- cycle ;
\draw [color={rgb, 255:red, 126; green, 211; blue, 33 }  ,draw opacity=1 ]   (142,68.33) -- (142,88.77) ;
\draw [color={rgb, 255:red, 74; green, 144; blue, 226 }  ,draw opacity=1 ]   (86.44,93.48) -- (137.27,93.48) ;
\draw  [color={rgb, 255:red, 100; green, 100; blue, 100 }  ,draw opacity=1 ][dash pattern={on 1pt off 1.5pt}] (15.5,27) -- (89.67,27) -- (89.67,102.33) -- (15.5,102.33) -- cycle ;
\draw  [color={rgb, 255:red, 100; green, 100; blue, 100 }  ,draw opacity=1 ][dash pattern={on 1pt off 1.5pt}] (104.97,27) -- (119,27) -- (119,71) -- (104.97,71) -- cycle ;
\draw  [color={rgb, 255:red, 100; green, 100; blue, 100 }  ,draw opacity=1 ][dash pattern={on 1pt off 1.5pt}] (165.6,27) -- (239.8,27) -- (239.8,101) -- (165.6,101) -- cycle ;
\draw  [color={rgb, 255:red, 100; green, 100; blue, 100 }  ,draw opacity=1 ][dash pattern={on 1pt off 1.5pt}] (135.2,57.36) -- (149.23,57.36) -- (149.23,101) -- (135.2,101) -- cycle ;

\draw (-1,29) node [anchor=north west][inner sep=0.75pt]  [font=\scriptsize]  {$b_{11}$};
\draw (195,105) node [anchor=north west][inner sep=0.75pt]  [font=\scriptsize]  {$r_{22}$};
\draw (241,88) node [anchor=north west][inner sep=0.75pt]  [font=\scriptsize]  {$b_{23}$};
\draw (241,60) node [anchor=north west][inner sep=0.75pt]  [font=\scriptsize]  {$a_{22}$};
\draw (241,29) node [anchor=north west][inner sep=0.75pt]  [font=\scriptsize]  {$b_{22}$};
\draw (196,16) node [anchor=north west][inner sep=0.75pt]  [font=\scriptsize]  {$r_{21}$};
\draw (178.94,60) node [anchor=north west][inner sep=0.75pt]  [font=\scriptsize]  {$a_{21}$};
\draw (164.42,14) node [anchor=north west][inner sep=0.75pt]  [font=\scriptsize]  {$b_{21}$};
\draw (165.63,103) node [anchor=north west][inner sep=0.75pt]  [font=\scriptsize]  {$b_{24}$};
\draw (136.29,105) node [anchor=north west][inner sep=0.75pt]  [font=\scriptsize]  {$c_{2}$};
\draw (133.73,44.82) node [anchor=north west][inner sep=0.75pt]  [font=\scriptsize]  {$r_{c_{2}}$};
\draw (107,72) node [anchor=north west][inner sep=0.75pt]  [font=\scriptsize]  {$r_{c_{1}}$};
\draw (105.77,16) node [anchor=north west][inner sep=0.75pt]  [font=\scriptsize]  {$c_{1}$};
\draw (73.7,14) node [anchor=north west][inner sep=0.75pt]  [font=\scriptsize]  {$b_{12}$};
\draw (75.51,103) node [anchor=north west][inner sep=0.75pt]  [font=\scriptsize]  {$b_{13}$};
\draw (56,60) node [anchor=north west][inner sep=0.75pt]  [font=\scriptsize]  {$a_{12}$};
\draw (43.16,16) node [anchor=north west][inner sep=0.75pt]  [font=\scriptsize]  {$r_{11}$};
\draw (44.67,105) node [anchor=north west][inner sep=0.75pt]  [font=\scriptsize]  {$r_{12}$};
\draw (-1,88) node [anchor=north west][inner sep=0.75pt]  [font=\scriptsize]  {$b_{14}$};
\draw (-1,60) node [anchor=north west][inner sep=0.75pt]  [font=\scriptsize]  {$a_{11}$};
\draw (14.17,115) node [anchor=north west][inner sep=0.75pt]  [font=\tiny]  {Variable Gadget for  $x_{1}$};
\draw (170.57,2) node [anchor=north west][inner sep=0.75pt]  [font=\tiny]  {Variable Gadget for $x_{2}$};
\draw (73.59,2) node [anchor=north west][inner sep=0.75pt]  [font=\tiny]  {Clause Gadget for  $c_{1}$};
\draw (123.19,115) node [anchor=north west][inner sep=0.75pt]  [font=\tiny]  {Clause Gadget for $ c_{2}$};

\end{tikzpicture}

    \caption{Example of reduction for $(\Bar{x}_1 \vee x_2)(x_1 \vee \Bar{x}_2)$}
    \label{fig:lsred}
\end{figure}

\section{Omitted Material from Section 4}

In this section, we provide an extensive discussion on adapting serial dictatorship to the roommate matching setting and highlight some natural ways to do this that fail to provide any guarantee on the social welfare.

\subsection{Manipulating \USW{} mechanisms}\label{app:manipulationexamples}

Our main result shows that social welfare and strategyproofness can co-exist under binary Leontief utilities. That is, there exist strategyproof mechanisms that always return a maximum welfare roommate matching. However, this is not trivial to achieve. 
We now present a maximum welfare mechanism which is not strategyproof for binary Leontief utilities. 

Consider the following parity based mechanism: given the reported instance $I$, compute the set of all maximum \USW{} roommate matchings: $\USW{}(I)$. Now if the number of matchings in $\USW{}(I)$ is odd, choose a matching from $\USW{}(I)$ where $i$ gets a utility of $1$ (under the reported valuations), else choose  a matching where $i$ gets utility $0$. Now clearly, if the set $\USW{}(I)$ has an even number of matchings, and $i$ has at least two matchings where $I$ gets a utility of $1$, $i$ can misreport and make the set of maximum \USW{} matchings under the reported instance have an odd number of matchings. As a result, the mechanism described is not strategyproof.

\subsection{Serial Dictatorship in Roommate Matchings} \label{app:sdessay}

Serial Dictatorship has been used to provide strategyproof mechanisms for a wide variety of models and applications \citep{elkind2021keeping,chan2016assignment}. In the roommate matching setting, a serial dictatorship mechanism proceeds as follows: An arbitrary ordering is chosen over the agents. In each iteration the next unmatched agent (according to the ordering) is called and they pick their favourite unmatched room and agent.  Agents have no control on when their turn is. When an agent's turn comes, they get to pick their favourite remaining bundle. Consequently, there is no incentive to lie under this mechanism.  

Prior work in the roommate matching setting has studied serial dictatorship only in the context of (exchange) stability \citep{chan2016assignment}. However, there exist bad picking sequences where serial dictatorship can conclude by matchings that are arbitrarily far from the optimal \USW{}. As a result, the typical serial dictatorship approach needs modification in order to guarantee a non-trivial approximation on \USW{}.

\paragraph{Naive Serial Dictatorship.}
The bottleneck to the standard serial dictatorship approach is due to agents who get no utility in any roommate matching under both additive and Leontief. Agents who have no value for any agent or room, can pick other agents who do, causing them to receive value zero in the matching. For example, consider the instance in \cref{fig:binaddSP} with an extra edge $(a_1,r_2)$, that is, $v_1(a_2) = \widehat{v}_1(r_2) =1$, $v_2(a_1)=v_2(a_3)= \widehat{v}_2(r_1) =1$, and all other values are zero. Suppose that the ordering of the agents is $a_4, a_3, a_2, a_1$ in the SD mechanism and $a_4$ chooses agent $a_2$ and room $r_2$. Then the roommate matching obtained is $\mu=\{(a_1,a_3,r_1),(a_2, a_4, r_2)\}$ and $\USW{}(\mu=0)$ under both additive and Leontief utilities. Whereas, the matching  $\{(a_2, a_3, r_1),(a_1,a_4,r_2)\}$ has \USW{} equal to $3$ and $1$ under additive and Leontief utilities, respectively.

\subsubsection{Unsuccessful Modifications} In order to prevent agents who are indifferent across all the available alternatives, we may consider different ways of changing the picking sequence. As it is not possible to get any non-trivial approximation to \USW{} for unrestricted additive utilities, we restrict our discussion to binary additive utilities. Analogous problems arise for binary Leontief utilities as well.

\paragraph{Approach 1.} A simple solution may be to order agents in decreasing order of how many rooms and agents they like. While this would improve the \USW{} bound (to \USW{$\sfrac{1}{6}$}), the mechanism would no longer be strategyproof. Agents prefer to pick first and would have an incentive to lie and say that they like more agents/rooms than they actually do.  In fact many approaches that picks agents based on their preferences are similarly not strategyproof for agents who are not indifferent to all (remaining) alternatives. 

This includes not just the ordering over agents, but also picking partners who like the same room. That is, an approach where the agent picking their partner chooses (from the set of unmatched agents they like) the agent who likes the same room as them, will not be strategyproof. Consider an agent $i$ to be the agent picking and $i$ likes agent $j$ and agent $k$. Now, agent $j$ likes $i$ and  knows that they have no value for any remaining agents apart from $i$, can lie and say they also like the same room as $i$ in order to be picked.
As a result, prioritizing agents for their preferences can often lead to incentives to misreport. 

\paragraph{Approach 2.}A crude modification to the standard serial dictatorship approach could simply pick an ordering over the agents who have {\em some} value, that is, they like at least one agent or room. After these agents have picked, the remaining agents are matched arbitrarily. This approach would not lead to a constant factor bound on the \USW{}.

Consider an instance with $6k+2$ agents and $3k+1$ rooms. Suppose the agents can be partitioned into four sets $A$, $B$, $C$ and $D$, i.e. $\N=A\cup B\cup C \cup D$. Now $A$ contains $2k$ agents labelled $a_{11},a_{12},a_{21},a_{22},\cdots, a_{k1},a_{k2}$. Similarly, $B$ contains  $2k$ agents as well, labelled $b_{11}, b_{12}, b_{21}, b_{22},\cdots, b_{k1},b_{k2}$ and $C$ also contains $2k$ agents labelled $c_{11}, c_{12}, c_{21}, c_{22},\cdots, c_{k1},c_{k2}$.  The set $D$ has the two remaining agents $d_1$ and $d_2$. We shall define the values for the rooms and agents so that a maximum welfare matching will matching $a_{i1}$ and $a_{i2}$ (as well as for the $b_i$s and $c_i$s) to the same room for each $i\in [k]$. In particular, for each $i\in [k]$, agent $b_{i1}$ only likes  $b_{i2}$ out of all the agents and vice versa. Similarly, for each $i\in [k]$, $c_{i1}$ and $c_{i2}$ only like each other. Further, agents $d_1$ and $d_2$ only like each other. Agents in $A$ have no value for any agent. 

We shall label the rooms according to who has value for them. We shall have one room $r^*$ and all the remaining $3k$ rooms will be labelled $r_{11}, r_{12}, r_{13}, r_{21},\cdots, r_{k1}, r_{k2},r_{k3}$. All agents in $A$ and $D$ have value $1$ for the room $r^*$ and value $0$ for all other rooms. For each $i\in [k]$ agent $b_{i1}$ has value $1$ for room $r_{i2}$ and agent $c_{i1}$ has value $1$ for room $r_{i3}$. All other room values for agents in $B$ and $C$ are $0$. 

Hence, it is easy to see that one maximum \USW{} matching will be $\mu=\{(d_1,d_2,r^*)\}$ $\bigcup_{i\in [k]}\{(a_{i1},a_{i2},r_{i1})\}$ $ \bigcup_{i\in [k]} \{(b_{i1},b_{i2},r_{i2})\}$ $ \bigcup_{i\in [k]}\{(c_{i1},c_{i2},r_{i3})\}.$ 
The \USW{} of $\mu$ is $6k+3$.

Now consider the picking sequence where $d_1$ picks first and then the agents in $A$. Agent $d_1$ will naturally choose the room $r^*$ and the agent $d_2$. This leaves all agents in $A$ with no remaining value. Now for each $i$ we can have agent $a_{i1}$ pick room $r_{i2}$ and agent $c_{i1}$ and agent $a_{i2}$ can pick room $r_{i3}$ and agent $b_{i1}$. These choices ensure that only the agents $d_1$ and $d_2$ get utility $2$ and $1$ respectively, and all other agents get utility $0$. For each $x>0$, there is a value of $k$ s.t. $\frac{1}{2k+1}<x$. As a result, this mechanism cannot guarantee any constant factor approximation on the \USW{}. A trivial bound that this does guarantee is \USW{$\frac{1}{n}$}.

This clearly shows that agents who have no value remaining must be discarded, which is exactly what the L/T Maximal Serial Dictatorship (\cref{alg:LT-alg}) and the Welfare-Prioritized Serial Dictatorship  mechanisms (\cref{alg:SDstar}) do. We now present the omitted proofs of its approximation on \USW{}.


\subsubsection{Serial Dictatorship for Binary Leontief Utilities}\label{app:sdleon}

We now demonstrate that the L/T Maximal Serial Dictatorship mechanism is effective in being strategyproof and proves to be \USW{$\sfrac{1}{6}$} for binary Leontief utilities,

\paragraph{Overview of the Mechanism.} The L/T Maximal Serial Dictatorship mechanism  picks an arbitrary ordering of agents and sequentially ``calls" them to pick their favourite room and agent pair out of those still unmatched. The only caveat is that if an agent has value $0$ for all remaining rooms and all remaining agents, they are not allowed to choose and set aside. Once all agents have been called, all the agents that were set aside are now arbitrarily assigned to the remaining rooms. 

We now show that an L/T maximal serial dictatorship mechanism proves to be strategyproof and \USW{$\sfrac{1}{6}$}.

\sdleon*

\begin{proof}
    The proof follows by showing that under binary Leontief utilities \cref{alg:LT-alg} is an L/T-maximal matching and implements a serial dictatorship mechanism.
     \paragraph{Welfare.}
     To prove the approximation bound of $\sfrac{1}{6}$, we show that \cref{alg:LT-alg} results in an L/T-maximal. Thus, by \cref{thm:LT-maximal}, any L/T-maximal matching is $\USW{\alpha}$ for $\alpha \geq \tfrac{1}{6}$ under binary Leontief utilities.
     Suppose for contradiction that at the end of the for loop, the matching is not L/T-maximal. This suggests that there exists at least one agent $i$ who can select a pair $(j, r)$ from remaining items, which contradicts the steps of the priority ordering in the for loop.

     \paragraph{Strategyproofness.}  
    Given a fixed priority ordering, at each step an unmatched agent $i$ picks an arbitrary agent-room pair $(j, r)$ among remaining only if $v_i(j) = 1$ and $\widehat{v}_{i}(r) = 1$.
    By \cref{obs:binleonlie} for binary Leontief utilities, agents cannot lie and increase the \USW{} of triples in which they truly get utility $1$. Similarly, an agent cannot decrease the \USW{} of a triple in which they truly get utility $0$. 
    Thus, no agent can benefit from misreporting its valuations.     
\end{proof}

\subsubsection{Serial Dictatorship for Binary Additive Utilities}\label{app:sdadd}
We now turn to additive utilities. As evident from \cref{thm:no-alpha-approx}, the $\USW{\sfrac{2}{3}}$ 
for additive utilities is not strategyproof. 
The approximation algorithms we have discussed so far and the \USW{$\sfrac{2}{3}$} in \cite{chan2016assignment} for maximizing social welfare fail to be strategyproof for additive utilities, even for the case of binary valuations. 
We now  demonstrate an extension of  the L/T maximal serial dictatorship mechanism (\cref{alg:LT-alg}) which we call Welfare-Prioritized Serial Dictatorship  (\cref{alg:SDstar}) which  is strategyproof and guarantees $\USW{\sfrac{1}{7}}$ under binary additive utilities. This works exactly the same as the L/T maximal mechanism. We call it by a different name only because Ls and Ts are not an exhaustive way of describing the varius possible triples for binary additive utilities.  

  
\begin{algorithm}[t]
  \KwIn{ Roommate Matching instance $ \langle \N,\R, v, \widehat{v} \rangle$ with binary valuations}
  \KwOut{A matching $\mu$}
  $\mu \gets \emptyset$\;
  Pick an arbitrary ordering on agents $\sigma$\;
  $P \gets \emptyset$ \Comment{set of agents with no value remaining}\;
  $Q\gets \N \cup \R$ \Comment{set of unmatched agents and rooms}\;
  \While{ $Q\cap \N \neq \emptyset$}{
    agent $i \in \N$ is the next agent in $\sigma$\;
    \eIf{$i$ has zero value for the remaining (agent,room) pairs }{
         $P\gets P \cup \{i\}$,
         $Q \gets Q \setminus \{i\}$\;
    }{
    
        Select $j_i \in \argmax _{j\in (\N\cap Q)} v_i(j)$\Comment{For tie braking, give preference to next agent on $\sigma$}\;
        Select $r_i \in \argmax _{r\in (\R\cap Q)} \widehat{v}_i(r)$\;
        $\mu \gets \mu \cup (i,j_i,r_i)$\;
        $Q\gets Q \setminus\{i,j_i,r_i\}$\;
    }
    }
    Match remaining unmatched agents in $P$ and unmatched rooms in $Q$ arbitrarily\;
    
   \caption{Welfare-Prioritized Serial Dictatorship} \label{alg:SDstar}
\end{algorithm}


 \sdstar*

\begin{proof}

We begin by proving that the mechanism is strategyproof.
    Observe that in the Welfare-Prioritized Serial Dictatorship mechanism, the only difference from the standard serial dictatorship mechanism is that it does not let agents who have no remaining value to choose. These agents are essentially indifferent about any feasible matching now possible, thus, there is no incentive for them to lie. For any other case, the incentive compatibility is inherited from that of the standard serial dictatorship mechanism.
    
    \begin{sloppypar}
    We now prove the guarantee on the \USW{}. Fix an arbitrary roommate matching instance $I=\langle N,R, \{\widehat{v}_i\}_{i\in N}, \{v_i\}_{i\in N} \rangle$. Let $\mu$ be the matching returned by \cref{alg:SDstar} on $I$. Let $\mu^*$ be a social welfare maximizing roommate matching on $I$.\end{sloppypar}

    Let $i$ be the first agent to choose their room $r_i$ and partner $j_i$. We need to bound the maximum social welfare that $i$'s choices can destroy in $\mu$ in comparison to $\mu^*$. With binary additive utilities, each agent can receive a value of $0$, $1$ or $2$ in any matching. Firstly, observe that $i$ will receive at least as much value in $\mu$ as they will in $\mu^*$. Thus, in the worst case, the values received by $i$ under the two matchings must be equal while $i$ pick an agent and a room that $i$ was not matched to under $\mu^*$.

    As in the proof of Lemma \ref{lem:sdleon}, beyond $i$ and $j_i$, at most four other agents can lose all possible value that they receive in $\mu^*$ by $i$'s choices. However, unlike the case of Leontief utilities, for additive utilities, the agents can still receive some value from one of their matched room or partner, even if not both. Thus, in the worst case, of the available agents and rooms after $i$ picking $j_i$ and $r_i$, these agents have no value left. If they were to, it would only improve the approximation achieved. Thus, they can only get value from $j_i$ or $r_i$, and thus, in total, these four agents can have a combined social welfare of $4$ in $\mu^*$.

    Now $j_i$ can lose a value of up to $2$ by being matched to $i$ and $r_i$. Thus, under $\mu^*$, these five agent can get a total value of $6$. In the worst case, $i$ receives a value of $1$ in both $\mu^*$ and $\mu$. Analogously, we can show that for any agent who picks their room and partner, they may destroy a social welfare of $6$, while getting only a value of $1$ for themselves. Hence, Welfare-Prioritized Serial Dictatorship (\cref{alg:SDstar}) gives a matching that achieves an approximation of $\sfrac{1}{7}$ on the maximum social welfare. 
    
\end{proof}

If the agent's compatibility values are symmetric, then the approximation bound is in fact better. 

\begin{corollary}\label{cor:sdstarsymm}
    When agent preferences are symmetric \cref{alg:SDstar} is \USW{$\sfrac{1}{6}$} for instances with binary additive utilities.
\end{corollary}

If the values are symmetric however, we can get a better bound. Recall that we say that agent values are symmetric if for all $i,j \in N$, $v_i(j)=v_j(i)$. In this case, if agent $i$ has value $1$ for $j_i$, then $j_i$ must also have value $1$ for agent $i$, thus, in this case, the worst approximation bound for these six agents would be $\sfrac{2}{7}$. 
    
Now if $i$ does not have value for $j_i$,  $i$ must not have value for any other agent either. In particular, $i$ must not have value for their partner under $\mu^*$. Let this agent be $j$. Consequently, $j$ also has no value for $i$. Further, recall that in the worst case, $j$ does not have value for any other room or agent that $j$ is not matched to under $\mu^*$. Thus, $j$ must receive no value under $\mu^*$. As a result, $i$'s choice can destroy a value of $2$ for $j_i$ and at most $3$ for the remaining agents. Consequently, \cref{alg:SDstar} gives a matching that achieves an approximation of $\sfrac{1}{6}$ on the optimal social welfare under binary symmetric additive utilities.

\end{document}